\documentclass[a4paper]{article}

\usepackage{inputenc, amsthm, amsmath} 
\usepackage{caption, subcaption, amssymb}
\usepackage{algpseudocode, algorithm}
\usepackage{calc, enumerate, wasysym}
\usepackage{graphicx, url}

\newlength{\mywidth}
\newcommand\bigfrown[2][\textstyle]
{\ensuremath{%
\addtolength{\arraycolsep}{-2mm}
\array[b]{c}\text{\resizebox{\mywidth}
{.9ex}{$#1\frown$}}\\[-1.3ex]#1#2\endarray}}
\newtheorem{theorem}{Theorem}[section]
\theoremstyle{definition}

\newtheorem{lemma}[theorem]{Lemma}
\newtheorem{cor}[theorem]{Corollary}

\newcounter{ResultsCounter}
\newcommand*{\email}[1]{\texttt{#1}}

\begin{document}
		\date{}	
		\title{A weighted binary average of point-normal pairs 
			   with application to subdivision schemes }
\author{Evgeny Lipovetsky\footnote{Corresponding author} $^{,}$\thanks{\email{evgenyl@post.tau.ac.il}, School of Computer Sciences, Tel-Aviv Univ.,Israel} \and 
Nira Dyn\thanks{\email{niradyn@post.tau.ac.il}, School of Mathematical Sciences, Tel-Aviv Univ., Israel}}

	\maketitle

%
%

		\begin{abstract}
			Subdivision is a well-known and established method for generating
			smooth curves and surfaces from discrete data by repeated refinements.
			The typical input for such a process is a mesh of vertices. 
			In this work we propose to refine 2D data consisting
			of vertices of a polygon and a normal at each vertex. Our core refinement procedure is based on a \textit{\textbf{circle average}}, which is a new non-linear weighted average of two points and their
			corresponding normals. The ability to locally approximate curves by 
			the circle average is demonstrated. With this ability, the circle
			average is a candidate for modifying linear subdivision schemes 
			refining points, to schemes refining point-normal pairs. This is done 
			by replacing the weighted binary arithmetic means in a linear
			subdivision scheme, expressed in terms of repeated binary averages, 
			by circle averages with the same weights. Here we investigate the 
			modified Lane-Riesenfeld algorithm and the 4-point scheme.
			For the case that the initial data consists of a control polygon only,
			a naive method for choosing initial normals is proposed. An example
			demonstrates the superiority of the above two modified schemes, with
			the naive choice of initial normals over the corresponding linear
			schemes, when applied to a control polygon with edges of significantly different lengths.	
		\end{abstract}		
	
\section{Introduction}
\label{sec:intro}
Subdivision schemes generate smooth curves/surfaces from discrete data
by repeated refinements. Linear schemes are well understood and have been used
in applications, such as Computer Graphics and Computer Aided Geometric Design.
The typical input to these schemes consists of a mesh of vertices.
For information on linear subdivision schemes see e.g. \cite{dl:02}.
In recent years linear schemes were adapted to refine other types of geometric
objects such as sets, manifold-valued data, and nets of functions (see e.g. \cite{df:02}, \cite{donoho}, \cite{wd:05}, \cite{cd:11}).

This paper is motivated by the idea to design subdivision schemes generating
surfaces by repeated refinements of 3D point-normal pairs. 
As a first step towards this aim we designed and investigated subdivision schemes
generating 2D curves by repeated refinements of 2D point-normal pairs (PNPs).
The subdivision schemes considered in this work are based on a geometric
construction. These schemes are significantly different from Hermite schemes,
which are linear schemes refining point-tangent pairs \cite{merrien:92}.
We plan to extend our schemes to schemes generating surfaces by refining point-normal pairs. It is important to note that point-normal pairs
can be obtained from point-tangent pairs but not vice-versa.

The approach taken here is similar to that taken in the adaptation of linear
subdivision schemes to manifold-valued data in \cite{ds:14},\cite{wd:05} and to 
sets in \cite{df:02},\cite{kd:13}.
The binary arithmetic mean in the refinement rules of linear subdivision 
schemes, expressed in terms of such repeated averages, is replaced by 
a weighted binary average of two PNPs. Such an average is designed here, 
based on a geometric construction involving a circle and hence its name
\textit{circle average}. 
With this average we modify the Lane-Riesenfeld algorithm~\cite{lr:80}, namely all
spline subdivision schemes, and the 4-point scheme \cite{dubucDeslauriers2},
\cite{dgl:87} to refine PNPs. 

Other modifications of these schemes which refine points
are available. The most relevant to our work are \cite{dosa:05},\cite{cahore:13},
and we plan to compare the performance of our modifications with their performance.

An interpolatory scheme refining PNPs, where the inserted PNP is determined by a similar
construction to the circle average with weight $\frac{1}{2}$, is presented in \cite{jue:07p}.
While in \cite{jue:07p}, the scheme converges and the limit of the normals
is equal to the normals of the limit curve, in our schemes this is not
necessarily the case. Yet our approach yields a variety of subdivision
schemes which are not limited to a subclass of initial PNPs as in
\cite{jue:07p}.

Here is an outline of the paper.

\noindent In section~\ref{sec:average} we first define the circle average by an explicit
geometric construction, and then prove that it is indeed an average. For that we prove 
the \textit{consistency property}, which guarantees that all repeated averages
originating from two PNPs can be expressed as one average with an appropriate weight.
We also
show that the circle average approximates well short pieces of smooth curves, 
which makes it a good candidate for modifying linear subdivision schemes refining
points to schemes refining PNPs, by the approach mentioned above. In 
section~\ref{sec:subdivision} we modify in
this way
 the Lane-Riesenfeld algorithm and also the interpolatory 4-point scheme. We prove that the modified schemes are convergent, and demonstrate
by figures and a video their editing capabilities. We provide also a simple
method for defining initial normals, in case the input consists
of control points only. The advantage of the resulting schemes over the
corresponding linear schemes is demonstrated for initial control polygons
with edges of significantly different lengths.

\section{The average}
\label{sec:average}
In this section we present the construction of a weighted binary average of
two pairs each consisting  of a point and a normal. All the weighted averages
of the two pairs are located on a circle.   
 When the two pairs are sampled from a circle, the weighted averages stay on that circle.
\subsection{Construction of the circle average}
We first introduce a new binary operation and then show that it is an average,
which we term \textit{the circle average}.
Given a real weight $\omega \in [0,1]$ and two pairs, each 
consisting of a point and a normal unit vector $P_0 = (p_0,n_0)$ and 
$P_1 = (p_1, n_1)$ in $2D$ space, we produce a new pair $P_\omega = 
(p_\omega, n_\omega)$ denoted by $P_0\circledcirc_\omega P_1$.
For $\omega = \frac{1}{2}$ we use also the shorter notation
$P_0\circledcirc P_1$.

To present the operation $P_0 \circledcirc_{\omega} P_1$ we introduce
some notation. 
The line defined by the vector $n_i$ and passing through 
the point $p_i$ is denoted by $l_i, i=0,1$.
The angle $\theta(u,v)$ denotes the angle between the vectors $u$ and $v$. 
In the special case of $u=n_0$ and $v=n_1$, the symbol $\theta$ 
substitutes $\theta(n_0,n_1)$. Observe that $0 \le \theta \le \pi$.
The length of the segment $[p_0,p_1]$ is denoted by $|p_0p_1|$,
and $\overrightarrow{p_0p_1}$ denotes the vector $\overrightarrow{p_1 - p_0}$.

Given three non-collinear points $a,b,c$, we denote by $bc$ the line
passing through $b$ and $c$, and by $HP(a;bc)$ 
the half-plane defined by the line $bc$ which contains the point $a$.
For two unit vectors $u=(\cos\alpha,\sin\alpha),v=(\cos\beta,\sin\beta)$,
we denote by $GA(u,v;\omega)$ their weighted geodesic average given by
\begin{align}
GA(u,v;\omega) = (\cos\gamma, \sin\gamma), \ 
\gamma = (1-\omega)\alpha + \omega \beta. \label{eq:geoavg}
\end{align}

\noindent The construction of $P_\omega = \{p_\omega, n_\omega\} = 
P_0\circledcirc_\omega P_1$ is done in several steps.
\begin{enumerate}
	\item Construct the perpendicular $[p_0,p_1]^{\perp}$ to the segment
	      $[p_0,p_1]$ at its midpoint.
	      Compute the angle $\theta$.
	      Construct two circles with centers $o_0$ and $o_1$ on
	      $[p_0,p_1]^{\perp}$, 
	      passing through $p_0$ and $p_1$, so that the central angles 
	      $\varangle p_{0} o_{i} p_{1}, i = 0,1$ equal $\theta$. Note that 
	      the two circles are symmetric relative to the segment $[p_0,p_1]$,
	      with the same radius $\frac{|p_0p_1|}{2sin\frac{\theta}{2}}$.
	\item For each circle, take the short arc connecting $p_0$ and $p_1$.
	      We call the above two arcs "candidate arcs", and the two circles
	      "candidate circles". One of the candidate arcs is chosen in the 
	      next step. We denote the selected candidate arc by 
	      $\bigfrown{P_0 \circledcirc P_1}$, its length by 
	      $|\bigfrown{P_0 \circledcirc P_1}|$, and the center of the 
	      corresponding circle by $o^*$.
	\item \textbf{Selection Criterion.} Let $q$ be the intersection point 
	      of $l_0$ and $l_1$. Consider the two half-planes defined by 
	      the line $p_0p_1$. If $n_0$ and $n_1$ are in different half-planes
	      (relative to $p_0p_1$) then take as $\bigfrown{P_0 \circledcirc P_1}$
	      the arc which is in the same half-plane as $q$,
	      otherwise $\bigfrown{P_0 \circledcirc P_1}$ is the other candidate 
	      arc. 
	\item Compute $p_\omega \in \bigfrown{P_0 \circledcirc P_1}$ such 
	      that the length of the part of $\bigfrown{P_0 \circledcirc P_1}$
	      between $p_0$ and $p_\omega$ is 
	      $\omega|\bigfrown{P_0 \circledcirc P_1}|$, or equivalently such that
	      the angle $\varangle p_0o^*p_\omega = \omega \theta$.
	\item Take the normal $n_\omega$ as $GA(n_0, n_1; \omega)$.
\end{enumerate}
See Figure~\ref{fig:construction} for examples. The selection criterion
and the following special cases are chosen to guarantee that the circle average
depends continuously on the data.
\\ \\
Special cases:
\begin{enumerate}[(i)]
	\vspace{-10px}
	\item If $\theta = 0$, i.e. $n_0 = n_1$, then $\bigfrown{P_0 
		  \circledcirc P_1} = [p_0,p_1]$, $p_\omega = (1-\omega)p_0 +
		  \omega p_1,$ and $n_\omega = n_0$.
	\item In case $\theta = \pi$ the construction is not defined.
	\item If $n_1 \parallel p_0p_1$ then 
	      we consider both normals to be in the same half-plane relative to $p_0p_1$, and $q$ to be in the same half-plane as $n_0$ 
	      when $\theta(n_1, \overrightarrow{p_0p_1}) = \pi$, and in the other
	      half-plane when $\theta(n_1, \overrightarrow{p_0p_1}) = 0$.
	      The case $n_0 \parallel p_0p_1$ is dealt with similarly.
	\item If $|p_0p_1| = 0$, i.e. $p_0 = p_1$, then $p_\omega = p_0$, and    
	      $n_\omega$ is computed as in 5.
\end{enumerate}
Note that $P_0 \circledcirc_0 P_1 = P_0$ and $P_0 \circledcirc_1 P_1 = P_1$.
\begin{figure} [!htb]
	\begin{subfigure}[b]{0.4\textwidth}
		\centering
		\includegraphics[scale=0.65]{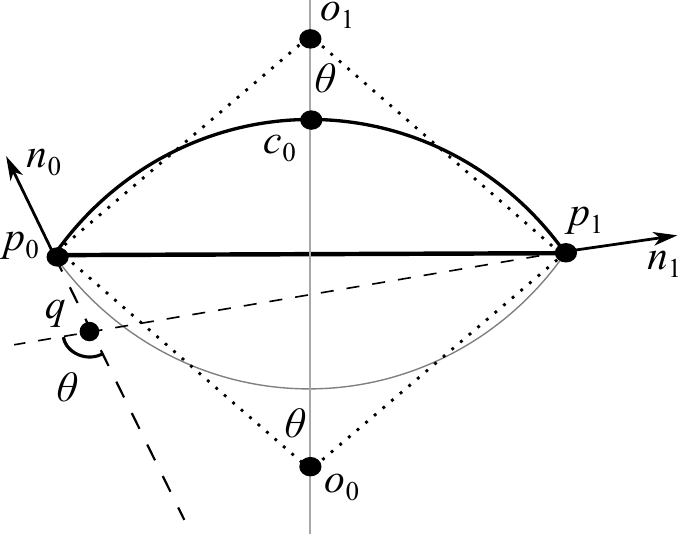} 
		\caption{$n_0$ and $n_1$ are in the same half-plane: 
		$p_{\frac{1}{2}} = c_0$.}
	\end{subfigure}
	\qquad\qquad\qquad
	\begin{subfigure}[b]{0.4\textwidth}
	\centering
		\includegraphics[scale=0.65]{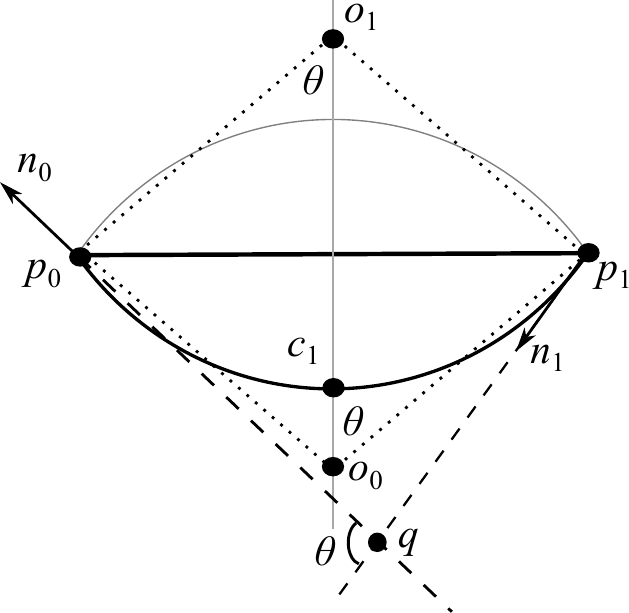} 
		\caption{$n_0$ and $n_1$ are in different half-planes: 
		$p_{\frac{1}{2}} = c_1$.}
	\end{subfigure}
	\captionsetup{justification=centering}
	\caption{Construction of $P_0 \circledcirc_{\frac{1}{2}} P_1$.
		\\$\bigfrown{P_0 \circledcirc P_1}$ is the bold arc.}
	\label{fig:construction}
\end{figure}

Two examples of the construction are given in Figure~\ref{fig:construction}.
In the left example, the point $c_0$ is taken as the point $p_\frac{1}{2}$ since 
$n_0, n_1 \notin HP(q;p_0p_1)$ and $c_1\notin HP(q;p_0p_1)$. 
In the right example, $n_1 \in HP(q;p_0p_1)$ while $n_0 \notin HP(q;p_0p_1)$.
Thus the point $c_1 \in HP(q;p_0p_1)$ is selected as $p_\frac{1}{2}$.
Note that the candidate arcs in both cases are the same, since in both examples
$\theta$ is the same.

In the next subsection we show that $P_0 \circledcirc_\omega P_1$
is indeed  a weighted average.


\subsection{The Consistency property}
In this section we show that 
\begin{align}
\forall t,s,k \in [0,1], 
(P_0 \circledcirc_t P_1) \circledcirc_k (P_0 \circledcirc_s P_1) 
= P_0 \circledcirc_{\omega^*} P_1,\  \omega^* = ks+(1-k)t \label{eq:consistency}
\end{align}
We call this property of the new operation 
\textit{\textbf{consistency}}. With this property the operation $\circledcirc_\omega$
is an average.

To prove (\ref{eq:consistency}), we first show
\begin{lemma}
	Assume w.l.o.g. that $t<s$.
	Let $P_t = P_0 \circledcirc_t P_1$, and $P_s = P_0 \circledcirc_s P_1$. 
	Then one of the candidate
	circles for $\bigfrown{P_t \circledcirc P_s}$ is the same as the 
	circle of $\bigfrown{P_0 \circledcirc P_1}$.
\label{lemma:createsamecircle}
\end{lemma}
\begin{proof}
	Let $o^*$ denote the center of the circle of $\bigfrown{P_0 \circledcirc P_1}$. 
	We show that this circle meets the requirements of a candidate circle
	for $\bigfrown{P_t \circledcirc P_s}$. 
	Indeed, it passes through $p_t$ and $p_s$, and the central angle
	$\varangle p_to^*p_s$ equals $(s-t)\theta$, which is the angle between $n_t$ 
	and $n_s$. Thus this circle is a candidate circle for 
	$\bigfrown{P_t \circledcirc P_s}$.
\end{proof}
Our proof of the consistency property is based upon a classical result
in Euclidean geometry.
\begin{lemma}
	Let $a,b,c,d$ be the four vertices of a convex quadrilateral and let 
	$\varangle a, \varangle b,$ $ \varangle c, \varangle d$ be the angles of 
	the quadrilateral at the corresponding vertices. Then
	\[\varangle a \ge \pi-\varangle b \iff \pi-\varangle d \ge \varangle c.\]
	See Figure~\ref{fig:keepsmallerangle} for an example.
\label{lemma:keepsmallerangle}
\end{lemma}
\begin{proof}
	Since $\varangle a + \varangle b + \varangle c + \varangle d = 2\pi$,
	$\varangle a + \varangle b \ge \pi \iff \varangle c + \varangle d \le \pi$,
	which proves the claim of the lemma.
\end{proof}
The preservation of inequality expressed in Figure~\ref{fig:keepsmallerangle}, follows
directly from the lemma.
\begin{figure} [!htb]
	\centering
	\begin{subfigure}[b]{0.4\textwidth}
		\includegraphics[scale=0.65]{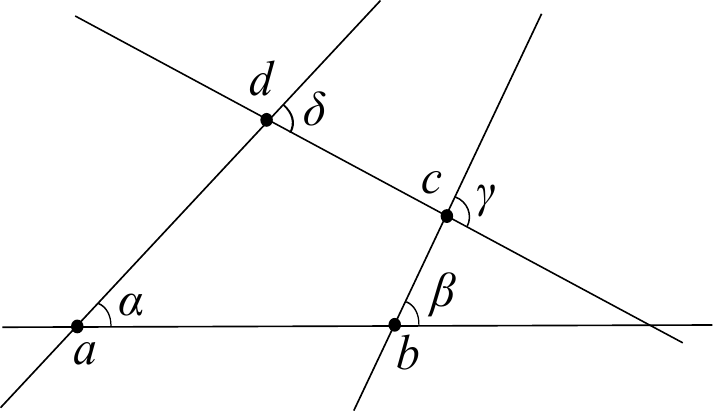} 
	\end{subfigure}
	\caption{Preserving the inequality $\alpha \le \beta \Rightarrow \delta\le \gamma$.}
	\label{fig:keepsmallerangle}
\end{figure}

Before proceeding we introduce more notation. Let $P_\omega = P_0
\circledcirc_\omega P_1 = (p_\omega, n_\omega)$. We denote by 
$l_\omega$ the line through
$p_\omega$ in direction $n_\omega$, and by $|\alpha_\omega|$ the angle
between the vectors $n_\omega$ and $\overrightarrow{p_0p_1}$. 
Note that $ \ 0 \le 
|\alpha_\omega| \le \pi.$ We introduce the convention that $\alpha_\omega > 0$
($\alpha_\omega < 0$) if $n_\omega$ is to the left (right) of
$\overrightarrow{p_0p_1}$, when both vectors are anchored in the same point.

We now prove the consistency property in case the two normals are in the same half-plane relative to $p_0p_1$ or equivalently that $\alpha_0\alpha_1 > 0$.
First, we show
\begin{figure} 
	\centering
		\includegraphics[scale=0.65]{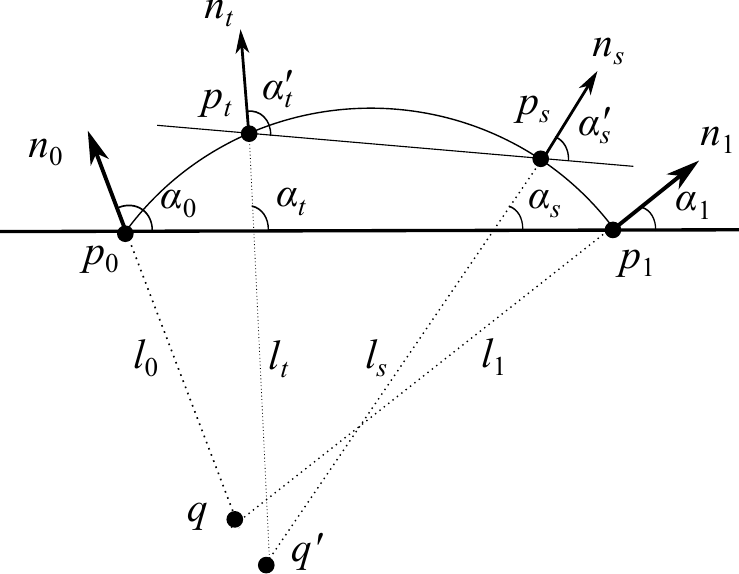} 
		\caption{The setup of Theorem~\ref{theorem:samenormhalfplane}.}
	\label{fig:consistency2}
\end{figure}
\begin{theorem}
	Let $n_0$ and $n_1$ be in the same half-plane relative to $p_0p_1$,
	and let $t,s \in [0,1]$, be such that $t < s$.
	Then,
	\[
	\bigfrown{P_t \circledcirc P_s} \subset \ \ \bigfrown{P_0 \circledcirc P_1} \ \ .\]  
	\label{theorem:samenormhalfplane}
\end{theorem}
\begin{proof}
	W.l.o.g., assume that $\alpha_0 > \alpha_1 > 0$ 
	(see Figure~\ref{fig:consistency2}).
	This assumption guarantees that $n_0, n_1 \notin HP(q;p_0p_1)$.
	Since the vectors $n_0$ and $n_1$ are in the same 
	half-plane relative to $p_0p_1$ the candidate arc in this half-plane 
	is selected by the selection criterion. 
	\\
	According to Lemma~\ref{lemma:createsamecircle}, 
	the circle containing \bigfrown{P_0 \circledcirc P_1} is considered as a candidate
	for $\bigfrown{P_t \circledcirc P_s}$. By definition of $n_t$
	\[ \alpha_t = (1-t)\alpha_0 + t\alpha_1, \ \alpha_s = (1-s)\alpha_0 + s\alpha_1 \]
	Since $t < s$ and $\alpha_0 > \alpha_1$, we obtain $\alpha_t > \alpha_s$.
	\\
	Let $\alpha'_t \ (\alpha'_s)$ be the angle between $p_tp_s$ and $l_t \ (l_s)$,
	and let $q'$ be the intersection point between $l_t$ and $l_s$.
	By Lemma~\ref{lemma:keepsmallerangle}, $\alpha_t > \alpha_s \Rightarrow \alpha'_t > \alpha'_s$. Therefore $n_t, n_s \notin HP(q';p_tp_s)$, implying that
	$\bigfrown{P_t \circledcirc P_s} \subset \ \ \bigfrown{P_0 \circledcirc P_1} \ \ $.
\end{proof}
To prove (\ref{eq:consistency}) it remains to show that for $P_t\circledcirc_k P_s =
(\tilde{p},\tilde{n})$, $\varangle \tilde{p}o^*p_0 = \omega^*\theta$, and 
$\theta(n_0, \tilde{n}) = \omega^*\theta$. Indeed
\[\varangle \tilde{p}o^*p_0 = \varangle p_to^*p_0 + k\varangle p_to^*p_s=
t\theta+k(s-t)\theta = \omega^*\theta,\]
and similarly
\[\theta(n_0,\tilde{n}) = \theta(n_0,n_t) + k\theta(n_t,n_s) = \omega^*\theta.\]


Next we discuss the case when the normals $n_0, n_1$ are 
in different half-planes relative to $p_0p_1$. 
\begin{figure}[!htb]
	\centering
	\includegraphics[scale=0.65]{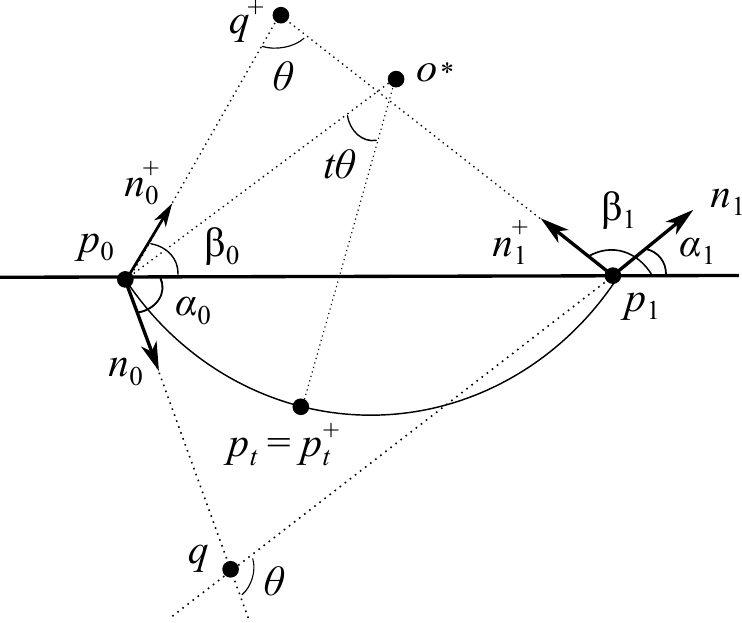} 
	\caption{The setup of Theorem~\ref{theorem:diffnormhalfplane}.}
	\label{fig:consistency3}
\end{figure}
\begin{theorem}
	For $n_0$ and $n_1$ in different half-planes relative to $p_0p_1$,
	the consistency, as defined in (\ref{eq:consistency}) holds.
	\label{theorem:diffnormhalfplane}
\end{theorem} 
\begin{proof}
	We assume w.l.o.g. that $n_0 \in HP(q;p_0p_1)$ and $n_1\notin HP(q;p_0p_1)$, 
	namely that $\alpha_0 <0,\ \alpha_1 > 0$, and that $\pi - |\alpha_0| > \alpha_1$ (see Figure~\ref{fig:consistency3}).	

	We take $\beta_0,\beta_1$ such that $0<\beta_0<\pi-\theta, $ and $\beta_1 = \beta_0 + \theta < \pi$, and define normal vector $n_i^+$ 
	such that $\theta(n_i^+,\overrightarrow{p_0p_1}) = \beta_i,$ and a pair
	$P_i^+ = (p_i, n_i^+),$ for $i = 0,1.$ Note that $\theta(n_0^+,n_1^+) = \theta = \theta(n_0,n_1)$.

	Let $q^+$ be the intersection point of the two lines defined for $i=0,1$
	by the vector $n_i^+$ and passing through the point $p_i$.
	By the choice of $\beta_0$ and $\beta_1$, we have $HP(q^+; p_0p_1) \neq HP(q;p_0p_1)$. Thus, according to the selection criterion, 
	$\bigfrown{P_0^+ \circledcirc P_1^+} = \ \ \bigfrown{P_0 \circledcirc P_1} \ \ $.

	Let $(p_\omega,n_\omega) = P_0 \circledcirc_\omega P_1$, and
	$(p_\omega^+,n_\omega^+) = P_0^+ \circledcirc_\omega P_1^+$. 
	By the definition of the average, $\varangle p_0o^{*}p_\omega = 
	\omega \theta, \varangle p_0o^{*}p_\omega^+ = 
	\omega \theta$, and since $p_\omega$ and $p_\omega^+$ are on $\bigfrown{P_0 \circledcirc P_1}$, they are equal.

	Let $0 <t <s< 1$. By the above discussion $p_t = p_t^+, \ p_s = p_s^+$,
	while 
	$n_t \neq n_t^+, n_s \neq n_s^+$. By (\ref{eq:geoavg}), we have
	\[\theta(n_t, n_s) = \theta(n_t^+, n_s^+) = (s-t)\theta.\]
	Thus, $\bigfrown{P_t \circledcirc P_s} = \ \ \bigfrown{P_t^+ \circledcirc P_s^+} \ \ $, and according to Theorem \ref{theorem:samenormhalfplane}
	\[ \bigfrown{P_t \circledcirc P_s} = \ \
	   \bigfrown{P_t^+ \circledcirc P_s^+} \ \ \subset
	   \ \ \ \ \bigfrown{P_0^+ \circledcirc P_1^+} \ \ \ \ = 
	   \ \ \ \ \ \ \bigfrown{P_0 \circledcirc P_1} \ \ \ \ \ \ .\]
	The rest of the proof of (\ref{eq:consistency}) is as in the proof of
	Theorem~\ref{theorem:samenormhalfplane}.
\end{proof}
Finally, we conclude from Theorem~\ref{theorem:samenormhalfplane} 
and Theorem~\ref{theorem:diffnormhalfplane},
\begin{cor}
	The consistency property holds regardless 
	of the location of the normals relative to $p_0p_1$.
\end{cor}
The consistency property of the operation $P_0 \circledcirc_\omega P_1, \ \omega 
\in [0,1]$ guarantees that it is a weighted binary average and allows to 
extend it for weights outside [0,1].
\\
Let $\omega^- <0$ and $\omega^+>1$. For $\omega^-$ we extend the arc
$\bigfrown{P_0 \circledcirc P_1}$ on the selected circle outward $p_0$, 
such that $\varangle p_{\omega^-}o^*p_0= |\omega^-|\theta$, and similarly, 
for $\omega^+$ we extend the arc outward $p_1$ such that $\varangle p_{\omega^+}o^*p_0
= \omega^+\theta$. The computation
of the normal is done by (\ref{eq:geoavg}).
See Figure \ref{fig:outerweights} for examples. 
\begin{figure}[!htb]
	\centering
	\includegraphics[scale=0.65]{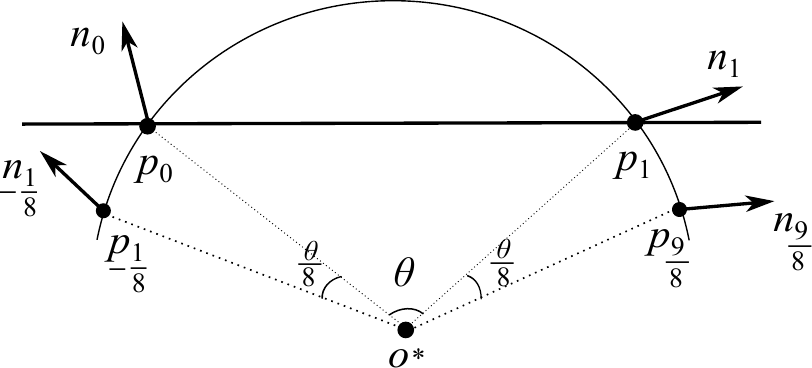} 
	\caption{Construction of the circle average with $\omega=-\frac{1}{8}$,
		$\frac{9}{8}$.}
	\label{fig:outerweights}
\end{figure}

It is easy to see that this extension is well defined for values of $\omega$ close to [0,1].

\pagebreak
\subsection{The arc $\bigfrown{P_0 \circledcirc P_1}$ as an approximation tool 
for curves}
In this subsection we compare the quality of the arc $P_0 \circledcirc P_1$
as an approximation tool for curves with that of the optimal arc approximating
curves in the least-squares sense. We expect the arc $\bigfrown{P_0 \circledcirc P_1}$ to approximate well short pieces of smooth curves, in analogy to the approximation capabilities of cubic Hermite interpolation \cite[Chapter~6]{Conte:1980:ENA:578374}.

Given a parametric curve $\Gamma(t)$, it is sampled at $\{t_i\}_{i=0}^{100}$
with $t_{i+1} - t_i = h>0$ and also its two normals $n_0$ and $n_{100}$ are 
sampled at $\Gamma(t_0)$ and $\Gamma(t_{100})$ . 
We solve the optimization problem of finding the circle $c_{opt}$ minimizing 
the sum of squares of distances 
to the input points. Next we construct the arc  $a_\Gamma = \bigfrown{P_0 \circledcirc P_{100}}$, where $P_i = (\Gamma(t_i), n_i)), i=0,100$. 
For every given $\Gamma(t_i)$, we find the nearest point on $c_{opt}$ and on 
$a_\Gamma$, and measure the distances to these points, denoted by
$\rho_i = dist(\Gamma(t_i), c_{opt}), \varrho_i = dist(\Gamma(t_i), a_\Gamma)$.
The next table presents values of two measures of the quality of the 
approximation  of three analytic curves by $c_{opt}$ and $a_\Gamma$ in two parametric intervals.
\begin{center}
	\begin{tabular}{ |c|c|c|c|c|c|c| } 
		\hline
		& & & & & & \\ 
		curve & $t_0$&$t_{100}$&
		$\max\limits_{0\leq i \leq100}{\rho_i}$& 
		$\max\limits_{0\leq i \leq100}{\varrho_i}$&
		$\frac{1}{101}\sum\limits_{i=0}^{100}\rho_i$ & 
		$\frac{1}{101}\sum\limits_{i=0}^{100}\varrho_i$\\ 
		& & & & & & \\ 
		\hline
		& & & & & & \\ 
$x(t) = \ 2\cos t $ & $\frac{5}{8}\pi$ & $\pi$ & 0.04145 & 0.05984 & 0.01315 & 0.02909 \\
		& & & & & & \\ 
$y(t)= \ \sin t $ & $\frac{12}{16}\pi$ & $\frac{15}{16}\pi$ & 0.00580 & 0.00710 & 0.00193 & 0.00377  \\
		& & & & & & \\ 
		\hline
		& & & & & & \\ 
$x(t)= t\cos t$   & $\frac{10}{8}\pi$ & $\frac{17}{8}\pi$ & 0.20098 & 0.28437 & 0.06613 & 0.14787 \\
		& & & & & & \\ 
$y(t) = t\sin t$ &$\frac{24}{16}\pi$ & $\frac{31}{16}\pi$  & 0.02337 & 0.02643 & 0.00794 & 0.01530 \\
		& & & & & & \\ 
		\hline
		& & & & & & \\ 

$x(t) = t^3-3t$ & 0 & $\frac{6}{8}\pi$ & 1.46814 & 1.97726 & 0.49597 & 1.02364\\	
		& & & & & & \\ 
$y(t) = t^2 - 1$&$\frac{3}{16}\pi$ & $\frac{9}{16}\pi$ & 0.25617 & 0.32838 & 0.09297 & 0.17556\\
		& & & & & & \\ 

		\hline
	\end{tabular}
	\captionof{table} {$c_{opt}$ vs. $a_\Gamma$}
	\label{tbl:opticirc}
\end{center}

The examples in Table~\ref{tbl:opticirc} demonstrate that $\bigfrown{P_0
\circledcirc P_1}$ can serve as an approximating tool in scenarios when the
sampling is expensive and/or when the computation time is critical. 
Moreover, the quality of the approximation by $\bigfrown{P_0 \circledcirc P_1}$
increases as the length of the interval of the parameter $t$ decreases.
This observation points to the advantage of approximating a curve by
piecewise arcs, and to the possibility of using the circle average in subdivision
schemes refining point-normal pairs.
\section{Subdivision schemes with circle averages}
\label{sec:subdivision}
In this section we consider subdivision schemes refining point-normal
pairs, which are obtained from converging linear subdivision schemes.
To obtain these schemes we express the linear schemes in terms of repeated
binary averages of points and replace these averages by the circle average.
We term the so obtained schemes "Modified schemes".

It is easy to verify that any modified subdivision scheme
reconstructs circles, namely, if the initial data is sampled from a circle, 
the limit of the modified scheme is that circle.

The convergence of the modified schemes is proved in two parts, the 
convergence of the points and the convergence of the normals. The proof of 
the convergence of the points is based on the following result:
\addtocounter{ResultsCounter}{1}
\paragraph{Result\ \Alph{ResultsCounter}}
(\cite{ds:14}, Theorem 3.6) A subdivision scheme refining points converges for any initial data, if any sequence 
of control polygons $\big\{ \mathcal{P}^j=\{p^j_i:i\in\mathbb{Z}\}\big\}_{j\in\mathbb{N}_0}$ generated
by this scheme satisfies
\begin{itemize}
	\item $e^{j+1} \leq \eta e^j$, $\eta \in (0,1)$, where $e^j$ is the maximal length
	of an edge in $\mathcal{P}^j$ (contractivity with factor $\eta$).
	\item $|p^{j+1}_{2i} - p^j_i| \leq ce^j$, with $c > 0$ (safe displacement).
\end{itemize}
Since our proof of convergence depends on the modified
subdivision scheme, it is given after the scheme is presented.

\subsection{The modified Lane-Riesenfeld (MLR) algorithm}
To obtain the first class of subdivision schemes we substitute 
the arithmetic average by the circle average in 
the linear Lane-Riesenfeld algorithm (LLR)\cite{lr:80}, obtaining the Modified Lane-Riesenfeld (MLR) algorithm, presented in Algorithm 1.
\begin{figure} 
	\vspace{-20px}
	\begin{subfigure}[b]{0.4\textwidth}
		\centering
		\includegraphics[scale=0.25]{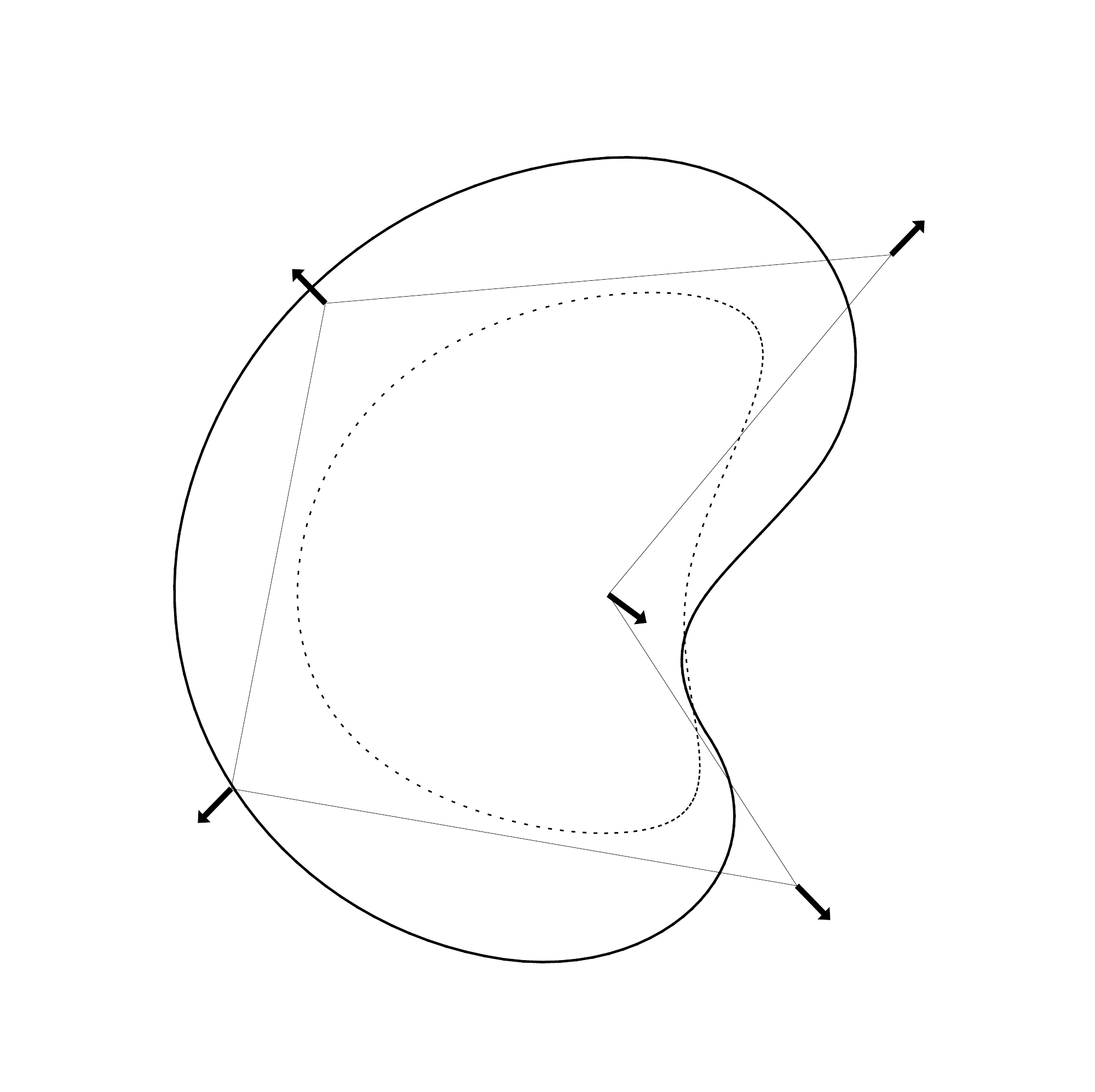} 
	\end{subfigure}
	\begin{subfigure}[b]{0.4\textwidth}
		\centering
		\includegraphics[scale=0.25]{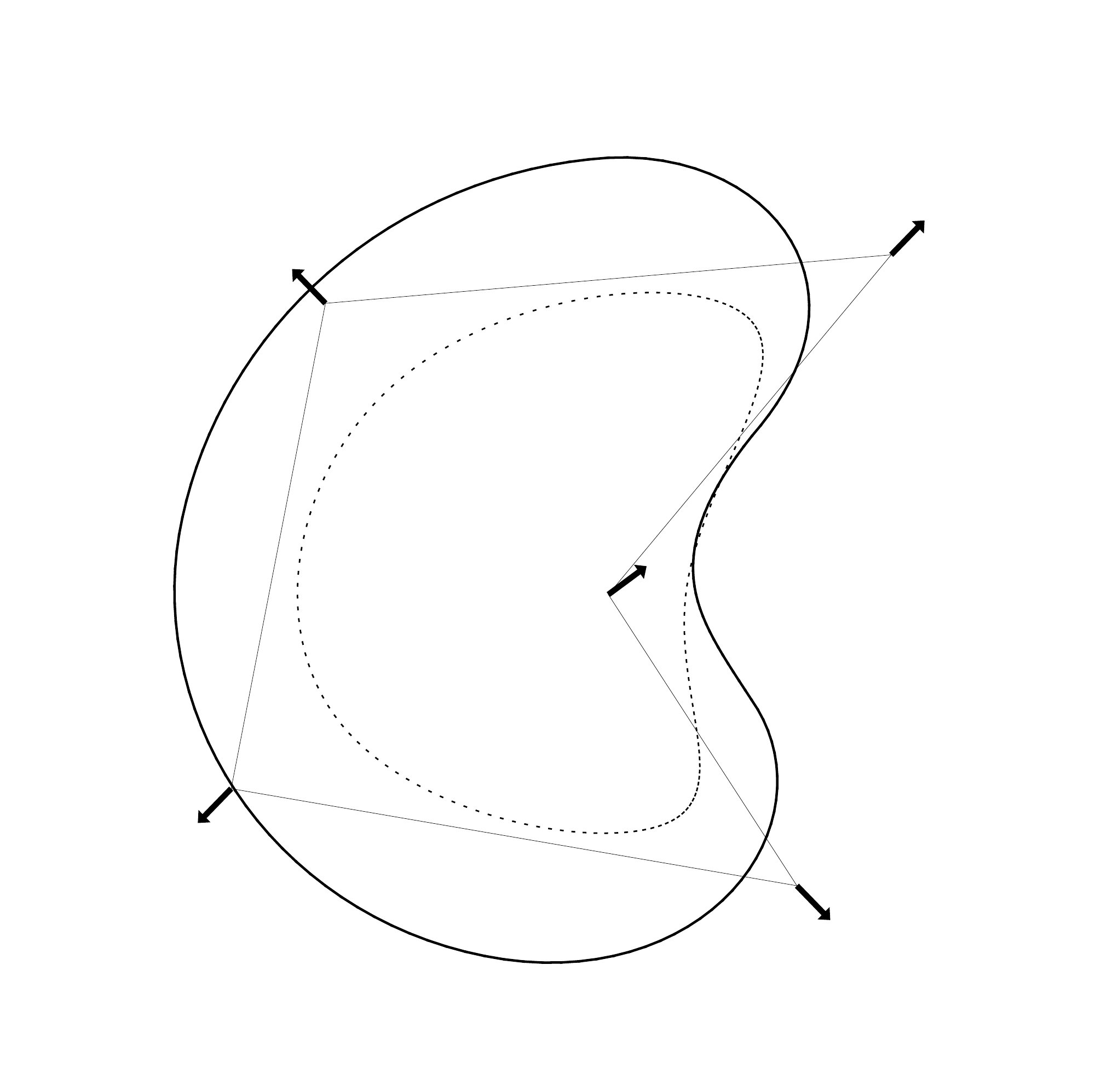} 
	\end{subfigure}
	\begin{subfigure}[b]{0.4\textwidth}
		\centering
		\includegraphics[scale=0.25]{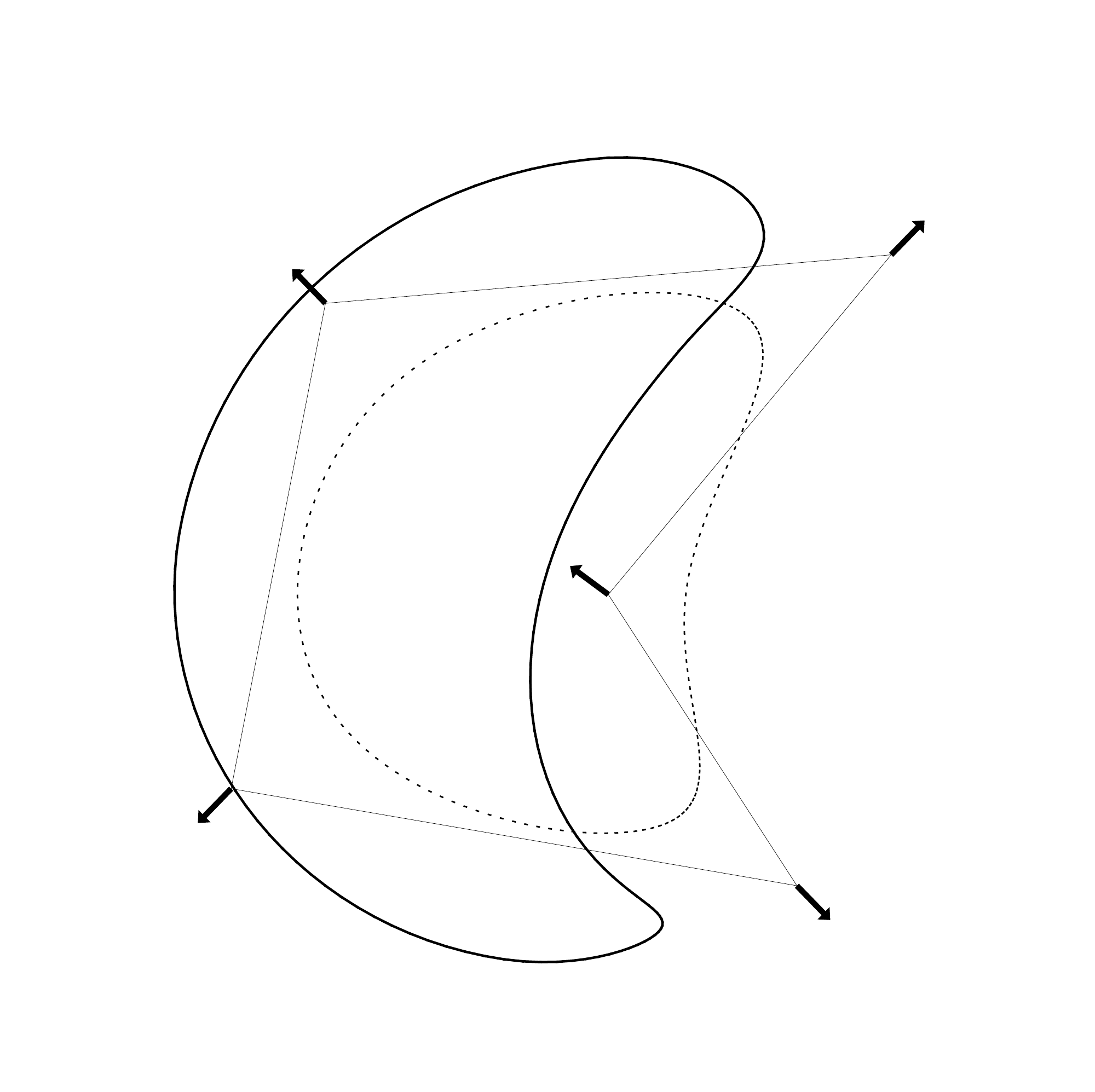} 
	\end{subfigure}
	\qquad\qquad\qquad
	\begin{subfigure}[b]{0.4\textwidth}
		\centering
		\includegraphics[scale=0.25]{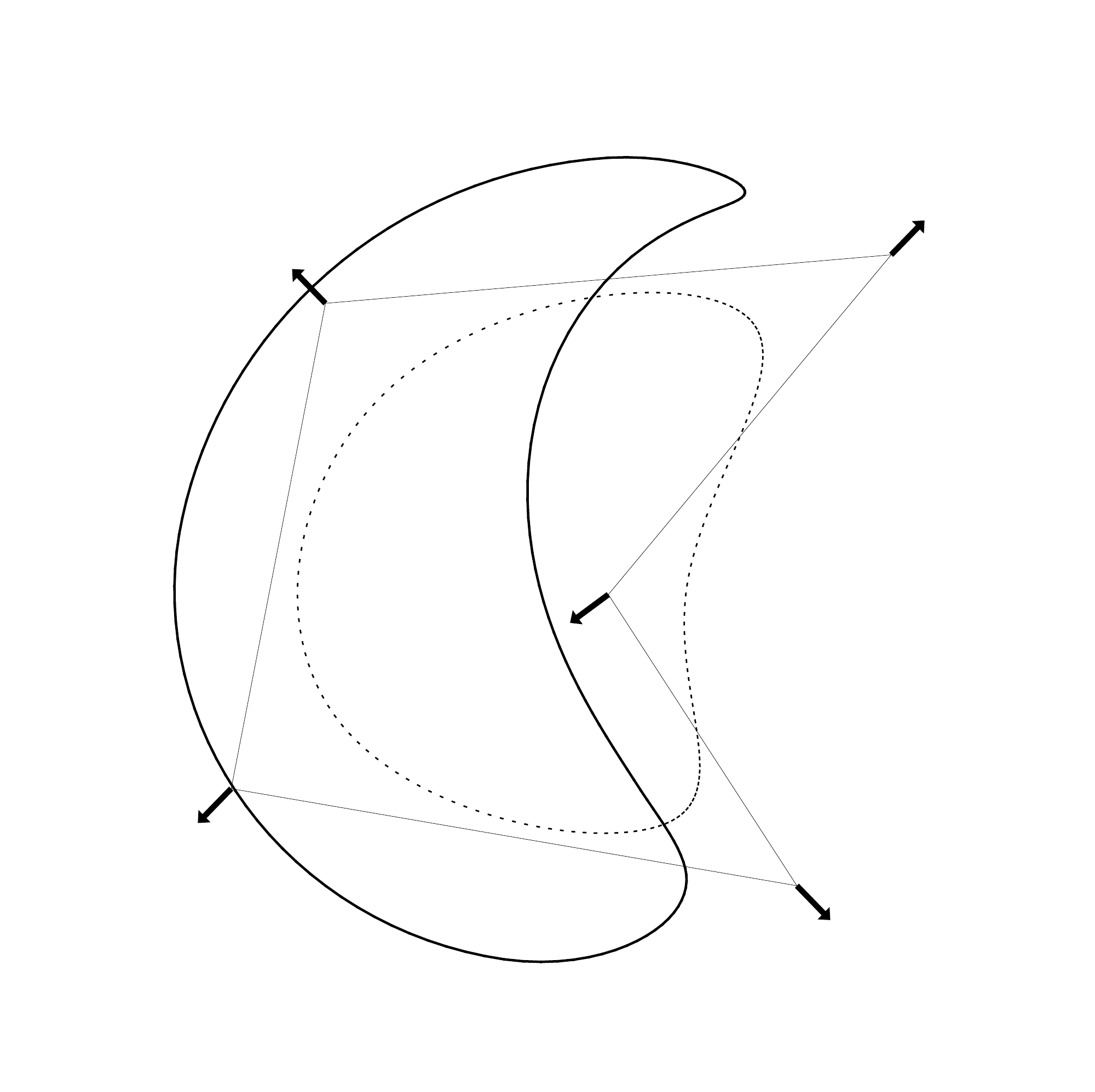} 
	\end{subfigure}
	\captionsetup{justification=centering}
	\caption{Editing capabilities of the MLR with $m=3$ by a change of 
		one initial normal. \\bold: MLR curve, dots: LLR curve.}
	\label{fig:mlr_changenorm}
\end{figure}
\begin{algorithm}
\caption{MLR}\label{algo:mlr}
\textbf{Input:} $m \in \mathbb{N}_0, \ P_i = (p_i,n_i),\ i\in \mathbb{Z}.$
\begin{algorithmic}
\For{$i \in \mathbb{Z}$}
	\State$P^0_i \gets P_i$
\EndFor
\For{j=1,2,\dots}
	\For{$i \in \mathbb{Z}$}
		\State $P^{j,0}_{2i}   \gets P^{j-1}_{i}$
		\State $P^{j,0}_{2i+1} \gets P^{j-1}_{i} \circledcirc_{\frac{1}{2}}
		 P^{j-1}_{i+1}$
        \rlap{\smash{$\left.\begin{array}{c@{}c@{}}\\{}\\{}\end{array}
        		       \right\}
		\begin{tabular}{l}elementary refinement\end{tabular}$}}
	\EndFor (i)
	\For{$k=1, \dots, m-1$}
		\For{$i \in \mathbb{Z}$}
			\State$P^{j,k}_{i} \gets P^{j,k-1}_{i} \circledcirc_{\frac{1}{2}} P^{j,k-1}_{i+1}$
            \rlap{\smash{$\left.\begin{array}{c@{}c@{}}\\{}\\{}\end{array}
            	\right\}
        		\begin{tabular}{l}smoothing step\end{tabular}$}}
		\EndFor (i)	
	\EndFor (k)
	\For{$i \in \mathbb{Z}$}
		\State$P^{j}_{i} \gets P^{j,m-1}_{i}$
        \rlap{\smash{$\left.\begin{array}{c@{}c@{}}\\{}\\{}\end{array}
       		\right\}
         		\begin{tabular}{l}result of current iteration\end{tabular}$}}
	\EndFor (i)
\EndFor (j)
\end{algorithmic}
\end{algorithm}

In Figure~\ref{fig:mlr_changenorm} we present curves generated by the MLR
algorithm with $m=3$ from the same initial data, but with one initial normal changed,
demonstrating the editing capabilities of the algorithm. For comparison we 
depict also the curves generated by the LLR algorithm.

\subsubsection{Convergence analysis}
First we prove the convergence of the points. Our analysis is based on 
Result A, which gives sufficient conditions for the convergence of a subdivision scheme refining points. These conditions in fact apply to any sequence of control polygons.

First, we introduce some additional notation related to the MLR algorithm.
For $k=0,...,m-1$ and $j\in\mathbb{N}_0$, $P^{j,k}_i = (p^{j,k}_i, n^{j,k}_i)$ 
and 
\begin{align}
\nonumber
e^{j,k} & = \max_{i\in\mathbb{Z}}{\{|p^{j,k}_i\ p^{j,k}_{i+1}|\}}, \\
\theta^{j,k} &= \max_{i\in\mathbb{Z}}{\{\theta(n^{j,k}_i,n^{j,k}_{i+1})\}}, \\
\nonumber
\mu^{j,k} &= \frac{1}{2\cos\frac{\theta^{j,k}}{4}}.
\label{eq:mlrdefs}
\end{align} 
We also define for $j\in\mathbb{N}_0$
\begin{align}
e^{j} = e^{j,m-1}, \ 
\theta^{j} = \theta^{j,m-1}, \
\mu^{j} = \mu^{j,m-1}.
\end{align}
Next we prove that the MLR satisfies the first condition of Result A
from a certain refinement level and on.
\begin{lemma}
	There exists $j^*\in \mathbb{N}_0$ such that the MLR algorithm is contractive in refinement levels above $j^*$,
	namely satisfies $e^{j+1} \leq \eta e^j$ with $\eta \in (0,1)$, for 
	$j \geq j^*$.
	\label{lemma:mlrcontractive}
\end{lemma}
\begin{figure} [!htb]
	\begin{subfigure}[b]{0.5\textwidth}
		\centering
		\includegraphics[scale=0.65]{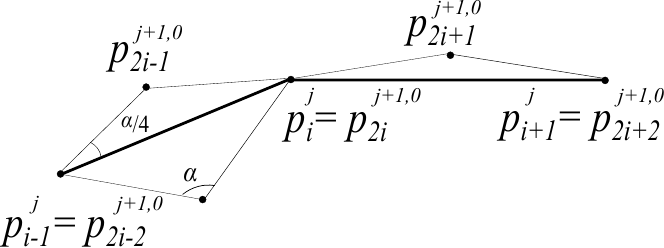} 
		\caption{$k=0$.}
	\end{subfigure}
	\begin{subfigure}[b]{0.5\textwidth}
		\centering
		\includegraphics[scale=0.65]{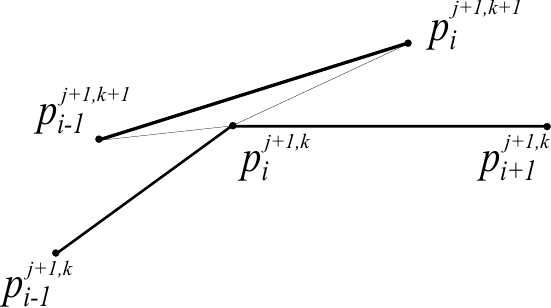} 
		\caption{$k=1,2,...,m-2$.}
	\end{subfigure}
	\captionsetup{justification=centering}
	\caption{The setup of Lemmas~\ref{lemma:mlrcontractive} and \ref{lemma:mlrdispsafe} }
	\label{fig:mlrconv}
\end{figure}
\begin{proof}
	Consider the pairs $\{P^{j,0}_{2i+1}\}_{ i \in \mathbb{Z}}$ inserted in 
	the elementary refinement step of the MLR algorithm. By the definition of 
	the circle average (see Figure~\ref{fig:mlrconv}a), we have
	\begin{align}
	|p^{j+1,0}_{2i}\ p^{j+1,0}_{2i+1}| = \frac{|p^{j}_i \ p^{j}_{i+1}|}{2\cos\Big(\frac{\theta(n^{j}_i,n^{j}_{i+1}) }{4}\Big)} \leq \frac{e^{j}}{2\cos\Big(\frac{\theta^{j}}{4}\Big)} 
	\leq \mu^{j} e^{j}.
	\label{ineq:insertbound}
	\end{align}
	Thus
	\begin{align}
	e^{j+1,0} \leq \mu^{j}e^{j}.
	\label{ineq:insertbound2}
	\end{align}
	In any smoothing step by the triangle
	inequality, and similar reasoning leading to (\ref{ineq:insertbound})
	(see Figure~\ref{fig:mlrconv}b), we have
	\[
	|p^{j,k+1}_{i}\ p^{j,k+1}_{i+1}|
	   \leq |p^{j,k+1}_{i}\ p^{j,k}_{i+1}| + |p^{j,k}_{i+1}\ p^{j,k+1}_{i+1}|
	   \leq 
	   \frac{e^{j,k}}{2\cos\Big(\frac{\theta^{j,k}}{4}\Big)} +
	   \frac{e^{j,k}}{2\cos\Big(\frac{\theta^{j,k}}{4}\Big)} 
	   \leq e^{j,k}\ 2\mu^{j,k}.
	\]
	Therefore
	\begin{align}
	e^{j,k+1}\leq e^{j,k}(2\mu^{j,k}), \ k = 0,\ldots,m-2.
	\label{ineq:nextsmoothiter}
	\end{align}
	Combining (\ref{ineq:nextsmoothiter}) and (\ref{ineq:insertbound2})
	we obtain
	\begin{align}
	\nonumber
	e^{j+1} = e^{j+1,m-1} \leq 2\mu^{j+1,m-2}e^{j+1,m-2} \leq \ldots 
	\nonumber
	\\ \leq (2\mu^{j+1, m-2})\ldots(2\mu^{j+1,0})e^{j+1,0}
	\nonumber
	\\ \leq \Big(\mu^{j} \prod_{k=0}^{m-2}{(2\mu^{j+1,k})} \Big)\ e^{j}
	\label{ineq:ejbound1}
	\end{align}
	Defining $\eta^{j+1} = \mu^{j} \prod_{k=0}^{m-2}{(2\mu^{j+1,k})}$ we obtain 
	from (\ref{ineq:ejbound1}) and (\ref{eq:mlrdefs})
	\begin{align}
	e^{j+1} \leq \eta^{j+1} e^{j},
	\label{ineq:ejbound2}
	\end{align}
	with
	\begin{align}
	\eta^{j+1} =
	\prod_{k=0}^{m-2}{\bigg(\frac{1}{\cos\frac{\theta^{j+1,k}}{4}}\bigg)}
	\frac{1}{2\cos\frac{\theta^{j}}{4}}
	\label{eq:etaj}
	\end{align}
	By the subdivision of the normals, we have
	\begin{align}
	\theta^{j+1,0} \leq \frac{1}{2}\theta^{j}, \ \
	\theta^{j+1,k} \leq\theta^{j+1,k-1}.
	\label{ineq:decrthetaj}
	\end{align}
	Thus
	\begin{align}
	\theta^{j+1} = \theta^{j+1,m-1} \leq \theta^{j+1,0}\leq\frac{1}{2}\theta^{j}
	\label{ineq:decrthetaj2}.
	\end{align}
	In view of (\ref{ineq:decrthetaj}) and (\ref{ineq:decrthetaj2})
	$\theta^{j,k}\leq\theta^{j,0}\leq\theta^{j-1}, k =0,...,m-1$, 
	and we get from
	(\ref{eq:etaj})
	\begin{align}
	\eta^{j+1}\leq \frac{1}{2}\bigg( \frac{1}{\cos\frac{\theta^{j}}{4}}\bigg)^m.
	\end{align}
	We also conclude from (\ref{ineq:decrthetaj2}) that
	$\frac{1}{\cos\frac{\theta^{j}}{4}}$ is monotone
	decreasing with $j$.
	\\
	Let $j^*$ be the minimal $j$ for which
	\begin{align}
	\bigg( \frac{1}{\cos\frac{\theta^j}{4}}\bigg)^m <2.
	\label{ineq:osthetajlessthan2}
	\end{align}
	Then for $j\geq j^*,\ \eta^j \leq \eta^{j^*}<1$ and by (\ref{ineq:ejbound2}) the MLR is
	contractive.
\end{proof}
Defining $\theta_m = {\theta^j}^*$ we obtain from (\ref{ineq:osthetajlessthan2})
\[
\theta_m = 4 \arccos \frac{1}{\sqrt[m]{2}}
\]
For $m=1$, $\theta_1 = 4\arccos\frac{1}{2} = 4\frac{\pi}{3} > \pi$, and since
the angle between any two normal vectors is at most $\pi$, we conclude 
that the MLR algorithm is contractive for any initial data from the first
level. Similarly for $m=2$, since 
\[\theta_2 = 4 \arccos \frac{1}{\sqrt{2}} = 4\frac{\pi}{4} = \pi.\]
For $m=3$, $\theta_3 = 4\arccos\frac{1}{\sqrt[3]{2}} > \frac{7}{9}\pi$
and the MLR algorithm with $m=3$ is contractive from level $j^*=1$.
We give in Table~\ref{tbl:mlrjstar} lower bounds of $\theta_m$ 
for several small values of $m$.
\begin{center}
	\begin{tabular}{ |c|c|c|c|c|c|c| } 
		\hline
		& & & & & & \\ 
		m & 1 &2&3&4&5&6\\ 
		\hline
		& & & & & & \\ 
		$\theta_m$ & $>\pi$ & $\ \pi$\ & $>\frac{7}{9}\pi$ &
		$>\frac{13}{18}\pi$ & $>\frac{11}{18}\pi$ & $>\frac{10}{18}\pi$ \\
		& & & & & & \\ 
		\hline
	\end{tabular}
	\captionof{table} {$\theta_m$ as a function of $m$}
	\label{tbl:mlrjstar}
\end{center}
As can be concluded from Table~\ref{tbl:mlrjstar}, $j^*=1$ for $3\leq m\leq6$.

To show the convergence of the MLR scheme by Result A, it remains to prove that
the scheme is displacement safe.
\begin{lemma}
	The MLR scheme is displacement safe.
	\label{lemma:mlrdispsafe}	
\end{lemma}
\begin{proof}
	The proof uses the notation of Lemma~\ref{lemma:mlrcontractive} and its
	proof. By the triangle inequality and since $p^{j+1,0}_{2i} = p^j_i$,\
	$p^{j+1}_{2i} = p^{j+1, m-1}_{2i}$ we get
	\begin{align}
		|p^{j+1}_{2i}\ p^j_i| \leq \sum\limits_{k=0}^{m-2}
		|p_{2i}^{j+1,k}p_{2i}^{j+1, k+1}|.
	\label{mlrdispsafe_ineq1}
	\end{align} 
	In view of Algorithm 1 and the geometry of the circle average (see 
	Figure~\ref{fig:mlrconv}) we have
	\[
	e^{j+1,k+1}\leq2\max\limits_{i}{|p^{j+1,k}_i\ p^{j+1,k+1}_i|},\ k=0,...,m-2,
	\] 
	\[
	|p^{j+1,k}_i\ p_i^{j+1,k+1}|\leq\frac{e^{j+1,k}}{2\cos\frac{\theta^{j+1,k}}{4}},
	\ k = 0,...,m-2,
	\]
	\[
	e^{j+1,0}\leq \frac{e^j}{2\cos\frac{\theta^j}{4}}.
	\]
	Thus for $k=0,...,m-2$,
	\begin{align}
	\begin{split}
	\max_{i}{|p_i^{j+1,k}\ p_i^{j+1,k+1}|} & \leq
	\frac{\max\limits_{i}{|p_i^{j+1,k-1}\ p_i^{j+1,k}|}}
	{\cos\frac{\theta^{j+1,k}}{4}}
	\leq \dots \leq \frac{\max\limits_{i}{|p_i^{j+1,0}\ p_i^{j+1,1}|}} {\prod\limits_{h=1}^{k}{\cos\frac{\theta^{j+1,h}}{4}}}
	\\
	&\leq \frac{e^{j+1,0}}{2\cos\frac{\theta^{j+1,0}}{4}\prod\limits_{h=1}^{k}
		{\cos\frac{\theta^{j+1,h}}{4}}}
	\leq \frac{e^j}{4\cos\frac{\theta^j}{4}\prod\limits_{h=0}^{k}
		{\cos\frac{\theta^{j+1,h}}{4}}}
	\end{split}
	\label{mlrdispsafe_ineq2}
	\end{align}
	By (\ref{ineq:decrthetaj}),(\ref{ineq:decrthetaj2}) and since
	$\theta^j\leq \theta^0 \leq \pi$ we have
	$\frac{\theta^{j+1,k}}{4} \leq \frac{\theta^j}{4} < \frac{\pi}{3},
	  \ k=0,1,\dots,m-1,$ 
	and (\ref{mlrdispsafe_ineq2}) can be replaced by
	\[
	\max_{i}{|p^{j+1,k}_{2i}\ p^{j+1,k+1}_{2i}|} \leq 
	\frac{e^j}{4\Big(\cos\frac{\pi}{3}\Big)^{k+2}}\leq 2^ke^j, 
	\ k = 0,1,\dots,m-2.
	\]
	Insertion of this bound in (\ref{mlrdispsafe_ineq1}) leads to
	\[
	\max_{i}{|p_{2i}^{j+1}p_i^j|}\leq e^j \sum_{k=0}^{m-2}{2^k} \leq 2^{m-1}e^j.
	\]
	This proves that the MLR scheme is displacement safe, with a constant
	which grows exponentially with $m$.
\end{proof}
We conclude from Lemmas~\ref{lemma:mlrcontractive}, \ref{lemma:mlrdispsafe} and 
Result A the convergence of the points. It remains to prove the convergence of 
the normals. Recalling that the operation  between the normals in the circle
average is a geodesic average independent of the points, the convergence of
the normals is a direct consequence of the following result, which is a special
case of Corollary 3.3 in ~\cite{ds:16}.
\paragraph{Result\ B}
(\cite{ds:16}, Corollary 3.3) The LR algorithm with the Euclidean average
replaced by a geodesic average is convergent.
\begin{cor}
	The MLR scheme for $m \ge 1$ is convergent.
\end{cor}

\subsubsection{Interactive demo}
\label{subsubsec:video}
We developed an interactive software with drawing capabilities, whose input
consists of point-normal pairs, and its output is the corresponding limit
of the MLR scheme with $m=1$, displayed on the screen. In this software
points can be dragged, normals can be rotated, and control
polygons can be extended and reflected. Also several control polygons can
be maintained simultaneously.  

As an example, the head of Mickey Mouse is drawn, starting from a simple
control polygon. A video of the drawing process, from an empty screen
to the final sketch of Mickey Mouse can be found at
\url{https://youtu.be/CGTiDztzVaM}.
This example demonstrates the drawing capabilities  of the MLR
scheme, with $m=1$, and the quality of naive choice of initial normals,
as explained in subsection~\ref{subsec:naive}.

\subsection{The modified 4-point scheme (M4Pt)}
In this section we modify the interpolatory linear 4-point subdivision 
scheme (L4Pt) \cite{dgl:87},\cite{dubucDeslauriers2},
\begin{align}
p^{j+1}_{2i} = p^{j}_{i}, \ 
p^{j+1}_{2i+1} = -\frac{1}{16}\big(p^{j}_{i-1} + p^{j}_{i+2}\big) + 
\frac{9}{16}\big(p^{j}_{i} + p^{j}_{i+1}\big)
\label{eq:l4pt}
\end{align}
We use the form suggested in \cite{kd:13} for the refinement rule in (\ref{eq:l4pt})
written in terms of repeated binary averages as
\begin{align}
p^{j+1}_{2i+1} = 
\frac{1}{2}\big(\frac{9}{8}p^{j}_{i} -\frac{1}{8} p^{j}_{i-1}\big) + 
\frac{1}{2}\big(\frac{9}{8}p^{j}_{i+1} - \frac{1}{8}p^{j}_{i+2}\big).
\label{eq:l4pt_repeated}
\end{align}
The modified 4-point scheme (M4Pt) with the circle average replacing 
the arithmetic average is presented in Algorithm 2. 

\begin{algorithm}
\caption{M4Pt}\label{algo:m4pt}
\textbf{Input:} $P_i = (p_i,n_i),\ i\in \mathbb{Z}.$
\begin{algorithmic}
	\For{$i \in \mathbb{Z}$}
		\State$P^0_i \gets P_i$
	\EndFor
	\For{j=1,2,\dots}
		\For{$i \in \mathbb{Z}$}
			\State $P^{j}_{2i}   \gets P^{j-1}_{i}$
			\State $S_L \gets 
			P^{j-1}_{i} \circledcirc_{-\frac{1}{8}} P^{j-1}_{i-1}$
			\State $S_R \gets 
			P^{j-1}_{i+1} \circledcirc_{-\frac{1}{8}} P^{j-1}_{i+2}$
			\State $P^{j}_{2i+1} \gets 
			  S_L \circledcirc_{\frac{1}{2}}S_R$			
		\EndFor
	\EndFor
\end{algorithmic}
\end{algorithm}
\pagebreak
Figure~\ref{fig:m4pt_changenorm} demonstrates the editing capabilities of the
M4Pt scheme by a change of one initial normal. Note that the control polygon 
and the normals in this example are the same as those in
Figure~\ref{fig:mlr_changenorm}.

\begin{figure} [!h]
	\begin{subfigure}[b]{0.4\textwidth}
		\centering
		\includegraphics[scale=0.25]{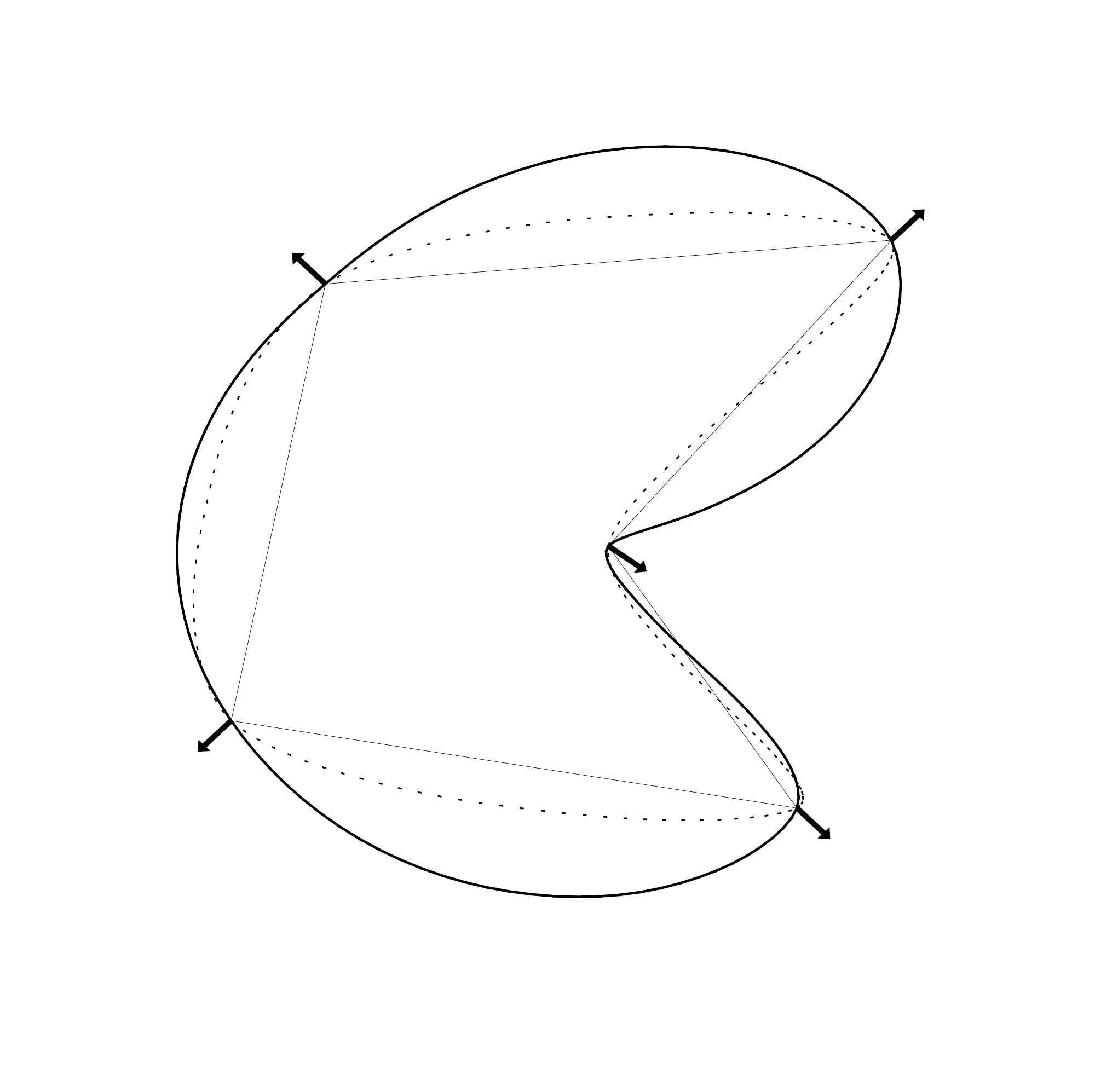} 
	\end{subfigure}
	\qquad\qquad\qquad\qquad
	\begin{subfigure}[b]{0.4\textwidth}
		\centering
		\includegraphics[scale=0.25]{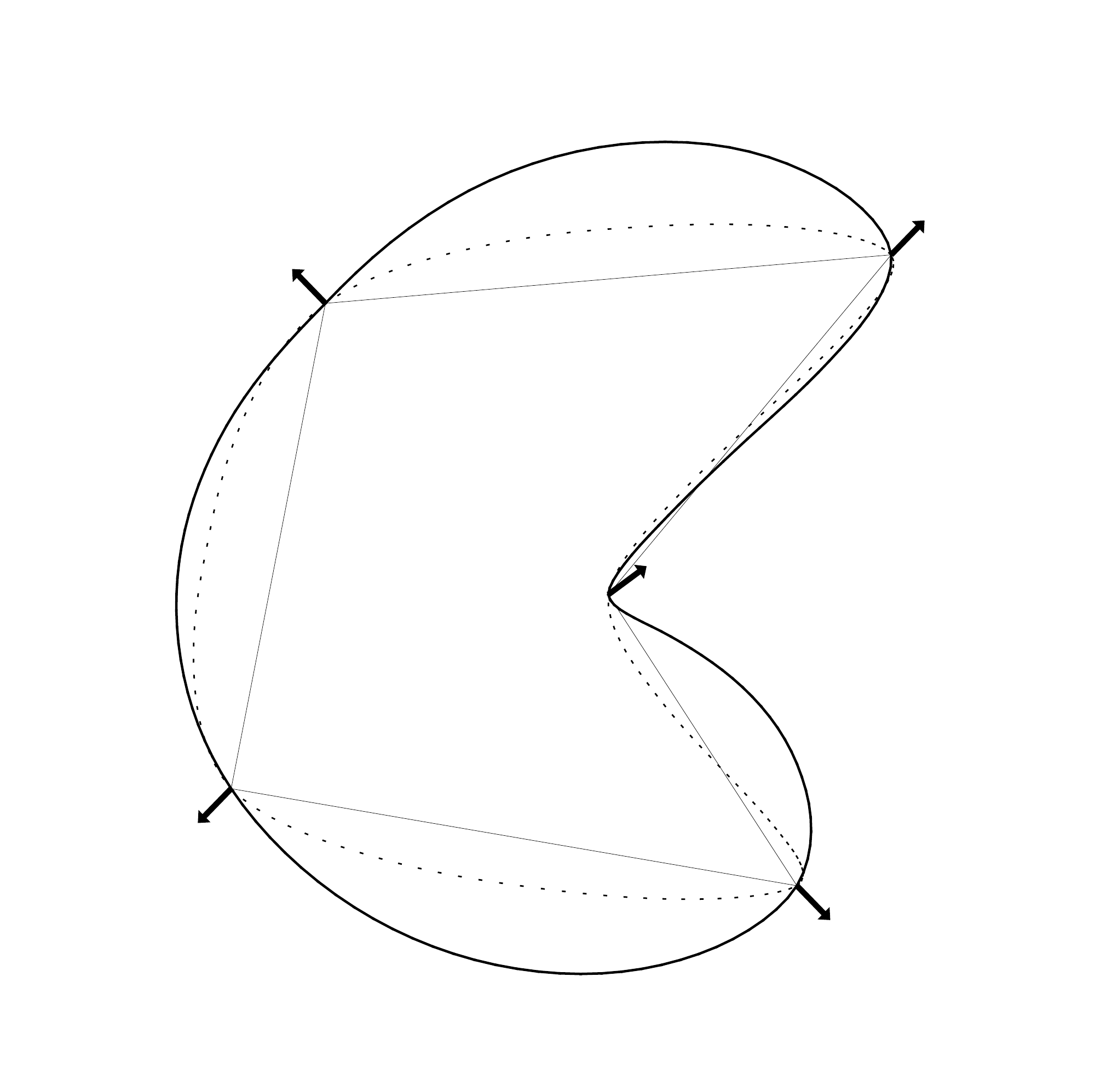} 
	\end{subfigure}
	\begin{subfigure}[b]{0.4\textwidth}
		\centering
		\includegraphics[scale=0.25]{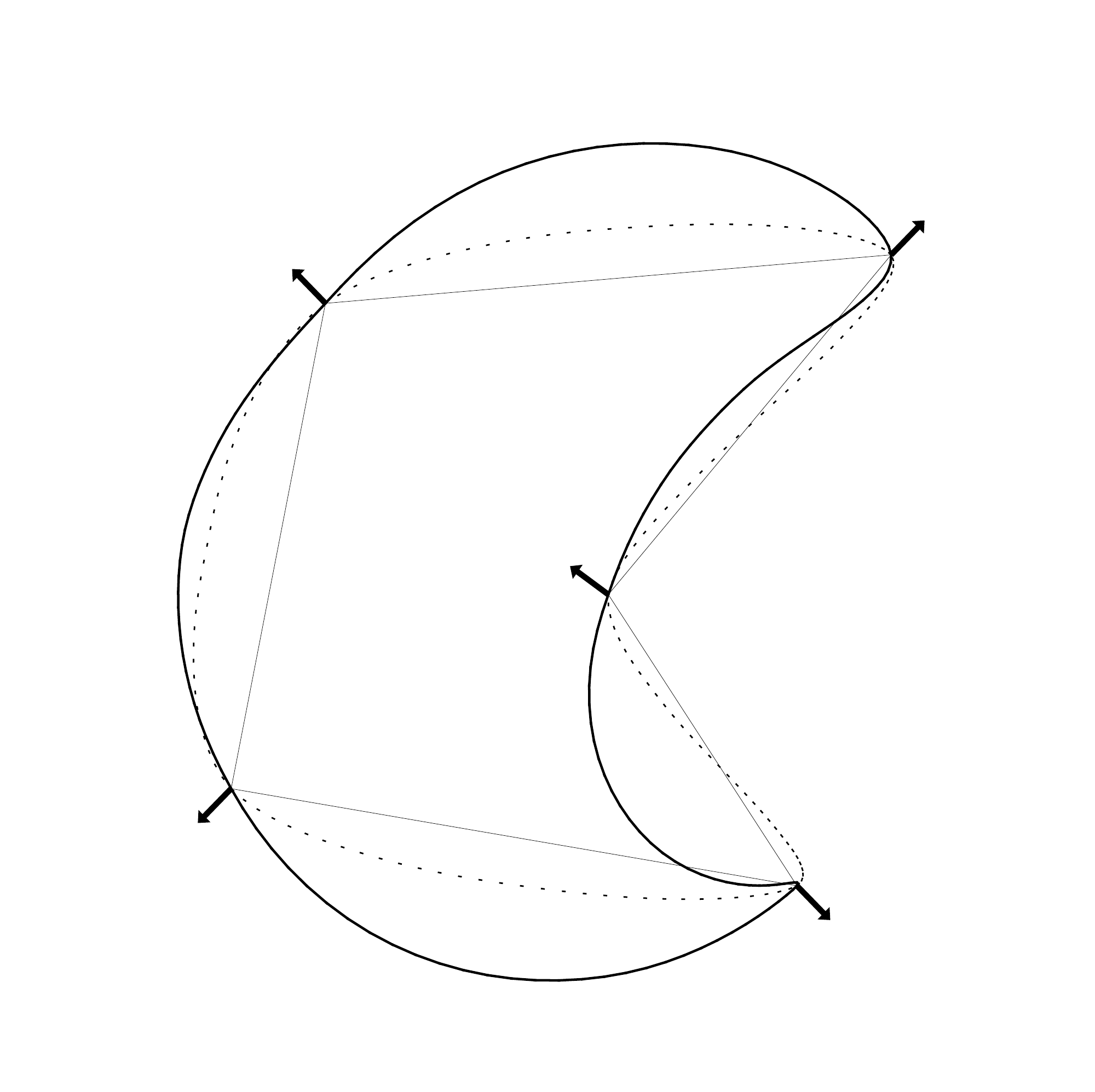} 
	\end{subfigure}
	\qquad\qquad\qquad\qquad
	\begin{subfigure}[b]{0.4\textwidth}
		\centering
		\includegraphics[scale=0.25]{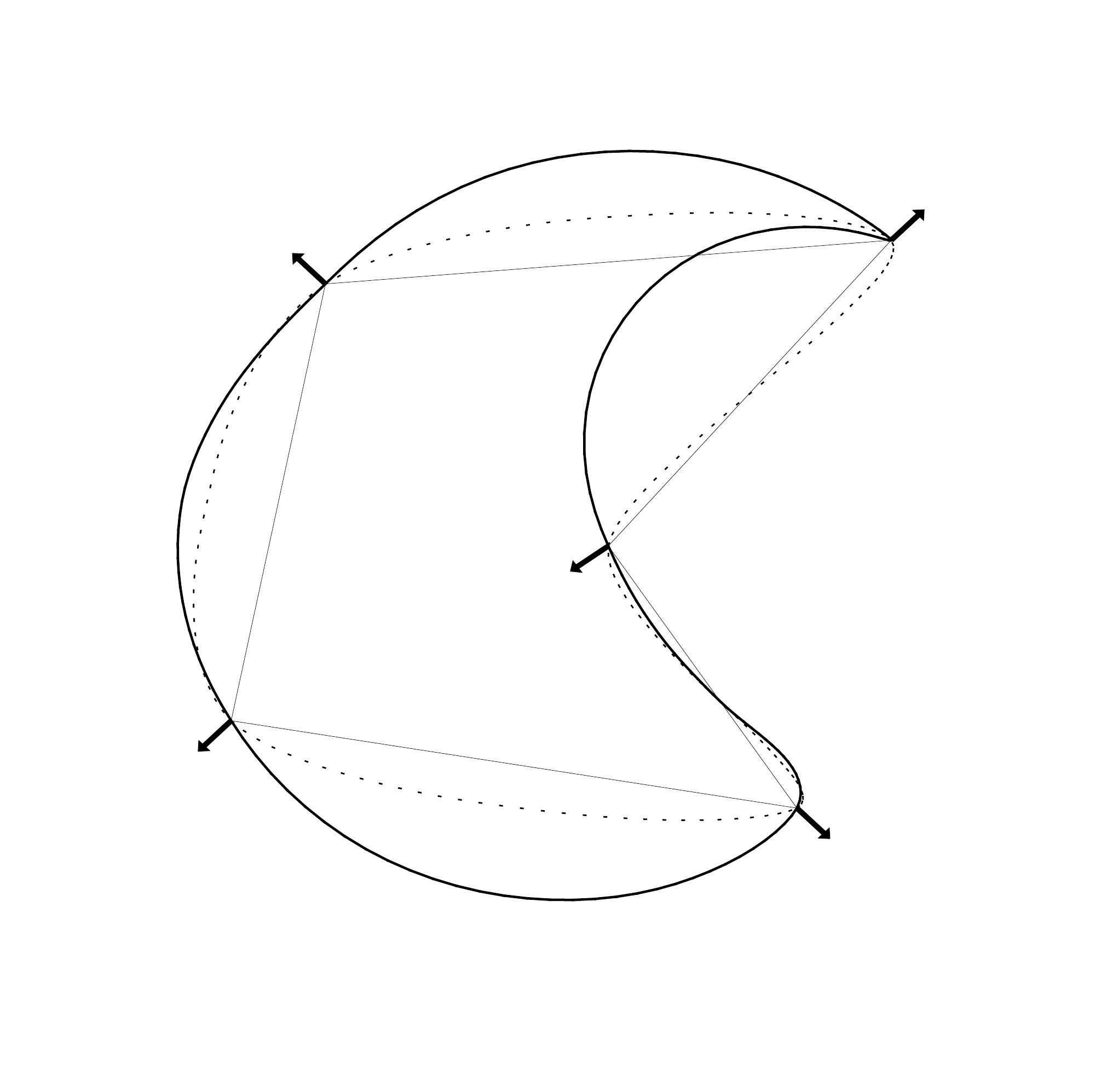} 
	\end{subfigure}
	\captionsetup{justification=centering}
	\caption{Editing capabilities of the M4Pt by a change of 
		one initial normal.\\bold: M4Pt curve, dots: L4Pt curve.}
	\label{fig:m4pt_changenorm}
\end{figure}

\subsubsection{Convergence analysis}
We begin the analysis by proving the convergence of the normals.
As we mentioned, the operation  between the normals in the circle
average is a geodesic average independent of the points. The convergence of
the normals is a direct consequence of the following result.
\paragraph{Result\ C}
(\cite{ds:14}, Example 5.1) The 4-point scheme adapted to manifold
valued data by replacing  in (\ref{eq:l4pt_repeated}) the average 
by geodesic average is convergent.
\\
\\
By definition, any interpolatory subdivision is displacement safe. Thus it remains to prove
the contractivity of the M4Pt, in order to show its convergence by Result A. 
\begin{lemma}{(Contractivity)}
	The M4Pt scheme is contractive for $j$ large enough.
\label{lemma:m4ptcontractivity}
\end{lemma}
\begin{figure} [!htb]
	\begin{subfigure}{1.0\textwidth}
		\centering
		\includegraphics[scale=0.65]{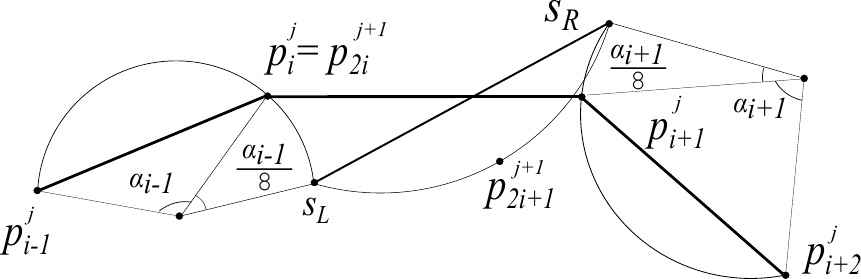} 
	\end{subfigure}
	\caption{The setup of Lemma~\ref{lemma:m4ptcontractivity}.}
	\label{fig:m4ptcontractivity}
\end{figure}
\begin{proof}
	Let $S_L = (s_L,n_L),\ S_R=(s_R, n_R)$ be the intermediate pairs 
	obtained by the M4Pt scheme (see Algorithm 2). For the proof we
	introduce the notation $\alpha_i=\theta(n^j_i,n^j_{i+1})$.
	By the triangle inequality and the geometry of the circle average
	(see Figure~\ref{fig:m4ptcontractivity}),
	\begin{align}
	|p_{2i}^{j+1} p_{2i+1}^{j+1}| \leq 
	|p_{2i+1}^{j+1} s_L|+|s_L p_{2i}^{j+1}| \leq
	\frac{|s_Rs_L|}{2\cos(\frac{1}{4}\theta(n_L,n_R))} + 
	\frac{|p^j_{i-1}p^j_i|sin\frac{\alpha_{i-1}}{16}}
	{sin\frac{\alpha_{i-1}}{2}}.
	\label{ineq:m4pt_c1}
	\end{align}	
	Next we show that
	\begin{align}
	\theta(n_L,n_R) \leq \frac{5}{4}\theta^j.
	\label{ineq:m4pt_c2}
	\end{align}
	Indeed, $\theta(n_L,n_R) \leq \theta(n_L,n^j_i) + \theta(n^j_i,n^j_{i+1})+
	\theta(n^j_{i+1}, n_R)$, where $\theta(n_L,n^j_i) = \frac{1}{8}
	\theta(n^j_{i-1},n^j_i)$ and similarly
	$\theta(n_R,n^j_{i+1}) = \frac{1}{8}\theta(n^j_{i+1},n^j_{i+2})$.
	Since $\theta^j = \max\limits_{i}{\theta(n^j_i,n^j_{i+1})}$, (\ref{ineq:m4pt_c2})
	follows.

	To bound $|s_L s_R|$ we use again the triangle inequality
	\[
	|s_L s_R| \leq |s_Lp^j_{i}| + |p^j_{i}p^j_{i+1}| + |p^j_{i+1}s_R|,
	\]
	and since $S_L = P^{j}_{i} \circledcirc_{-\frac{1}{8}} P^{j}_{i-1}, \ 
	S_R =
	P^{j}_{i+1} \circledcirc_{-\frac{1}{8}} P^{j}_{i+2}$,
	\begin{align}
	|s_Lp^j_i| \leq \frac{|p^j_{i-1}p^j_i| sin\frac{\alpha_{i-1}}{16}}
	{sin\frac{\alpha_{i-1}}{2}}, \ 
	|s_Rp^j_{i+1}| \leq \frac{|p^j_{i+2}p^j_{i+1}| sin\frac{\alpha_{i+1}}{16}}
	{sin\frac{\alpha_{i+1}}{2}}
	\end{align}
	Thus
	\begin{align}
	|s_Ls_R| \leq e^j\Big(1 + \frac{sin\frac{\alpha_{i-1}}{16}}
	                      {sin\frac{\alpha_{i-1}}{2}} + 
	                      \frac{sin\frac{\alpha_{i+1}}{16}}
	                      {sin\frac{\alpha_{i+1}}{2}}\Big),
	\label{ineq:m4pt_slsr} 
	\end{align}
	and we get from (\ref{ineq:m4pt_c1}), (\ref{ineq:m4pt_c2}) and 
	(\ref{ineq:m4pt_slsr})
	\[
	|p^{j+1}_{2i} p^{j+1}_{2i+1}| \leq e^j\Big(1 + \frac{sin\frac{\alpha_{i-1}}{16}}
	{sin\frac{\alpha_{i-1}}{2}} + 
	\frac{sin\frac{\alpha_{i+1}}{16}}
	{sin\frac{\alpha_{i+1}}{2}}\Big)\frac{1}{2\cos\frac{5}{16}\theta^j} +
		e^j\Big(\frac{sin\frac{\alpha_{i-1}}{16}}
		{sin\frac{\alpha_{i-1}}{2}}\Big).
	\]
	Similarly
	\[
	|p^{j+1}_{2i+1} p^{j+1}_{2i+2}| \leq e^j\Big(1 + \frac{sin\frac{\alpha_{i-1}}{16}}
	{sin\frac{\alpha_{i-1}}{2}} + 
	\frac{sin\frac{\alpha_{i+1}}{16}}
	{sin\frac{\alpha_{i+1}}{2}}\Big)\frac{1}{2\cos\frac{5}{16}\theta^j} +
	e^j\Big(\frac{sin\frac{\alpha_{i+1}}{16}}
	{sin\frac{\alpha_{i+1}}{2}}\Big).
	\]
	Therefore
	\[
	e^{j+1} \leq e^j\Big(1 + \frac{sin\frac{\alpha_{i-1}}{16}}
	{sin\frac{\alpha_{i-1}}{2}} + 
	\frac{sin\frac{\alpha_{i+1}}{16}}
	{sin\frac{\alpha_{i+1}}{2}}\Big)\frac{1}{2\cos\frac{5}{16}\theta^j} +
	Ae^j,
	\]
	with 
	\[A=\max_{i}{\frac{sin\frac{\alpha_{i}}{16}}
		{sin\frac{\alpha_{i}}{2}}}.
	\]
	Thus, $e^{j+1} \leq \eta^je^{j}$ with
	\begin{align}
	\eta^j = \frac{1}{2}\Big(1 + \frac{sin\frac{\alpha_{i-1}}{16}}
	{sin\frac{\alpha_{i-1}}{2}} + 
	\frac{sin\frac{\alpha_{i+1}}{16}}
	{sin\frac{\alpha_{i+1}}{2}}\Big)\frac{1}{\cos\frac{5}{16}\theta^j} +
	A.
	\label{ineq:m4pt_etaj}
	\end{align}
	Since $\alpha_i \leq \theta^j$ and the normals converge, 
	$ \lim_{j\rightarrow\infty}{\theta^j} = 0 $. Thus we get from 
	(\ref{ineq:m4pt_etaj})
	\begin{align}
	\eta^* = \lim_{j\rightarrow\infty}{\eta^j} 
	= \frac{1}{2}(1+\frac{1}{8}+\frac{1}{8}) + \frac{1}{8} = \frac{3}{4}
	\label{ineq:m4pt_etastar}
	\end{align}
	We conclude from (\ref{ineq:m4pt_etastar}) that for $j$ large enough
	$\eta^j<1$. Defining $J^*$ such that $\eta^j<\frac{7}{8}$ for $j \geq J^*$,
	we get that the M4Pt scheme is contractive for $j \geq J^*$, with
	$\eta^j = \frac{7}{8}$.
\end{proof}
We conclude from Lemma~\ref{lemma:m4ptcontractivity} and Result A 
the convergence of the points.
\begin{cor}
	The M4Pt scheme is convergent.
\end{cor}

\subsection{Naive choice of initial normals}
\label{subsec:naive}
In previous sections we discussed the scenario in which normals are given
at every vertex of the input control polygon. In this section we propose a 
method for determining initial normals at the vertices of a given control polygon.
\begin{figure} [!htb]
	\begin{subfigure}[b]{1.0\textwidth}
		{(a)\hspace{17px} } \includegraphics[scale=0.8]{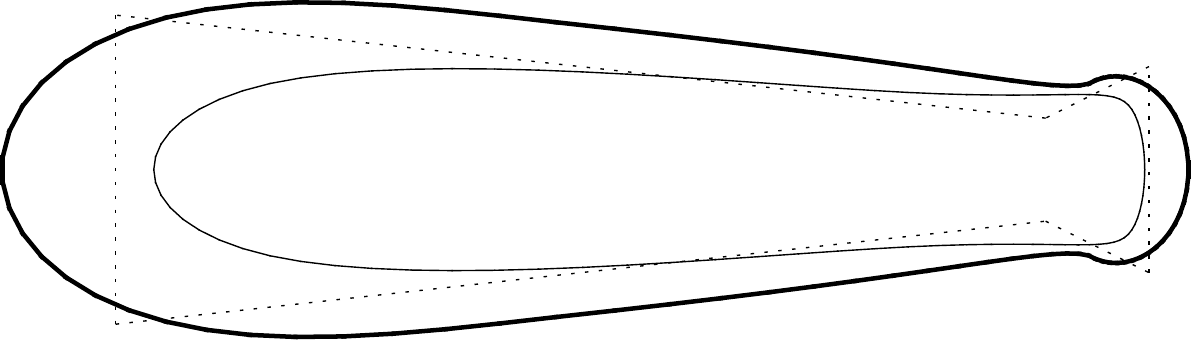}
	\end{subfigure}
	\\ \\
	\begin{subfigure}[b]{1.0\textwidth}
		{(b)\hspace{22px}}\includegraphics[scale=0.8]{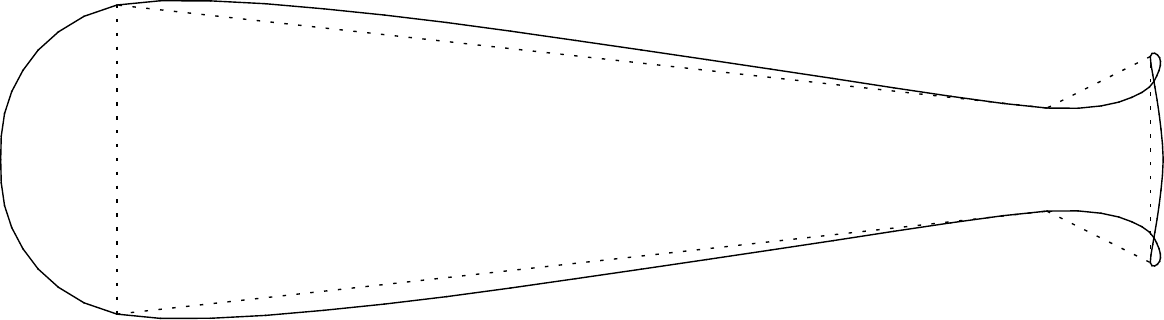}
	\end{subfigure}
	\\ \\
	\begin{subfigure}[b]{1.0\textwidth}
		{(c)}\includegraphics[scale=0.8]{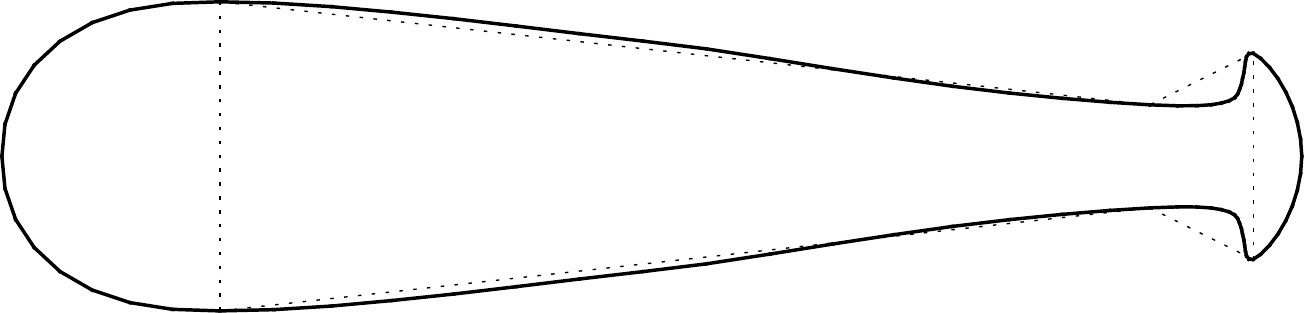}
	\end{subfigure}
	\captionsetup{justification=centering}
	\caption{Comparison between modified schemes and their corresponding 
		linear schemes.
		Same initial control polygon (dots); (a)MLR (bold) and LLR (regular); (b)L4Pt; (c) M4Pt. }
	\label{fig:bottle_naivenorm}
\end{figure}

To determine a normal at the vertex $p_i$, we first
compute the normals  $v_{i-1}, v_i$ to the neighboring edges of the vertex $p_i$,
and the length of these edges, $d_{i-1}, d_i$. We chose the direction of $v_i$,
the normal to the edge $p_ip_{i+1}$, such that $v_i \ \times \ \overrightarrow{p_ip_{i+1}} > 0$. The normal at $p_i$ is the weighted
geodesic average of $v_{i-1},v_i$ as defined in (\ref{eq:geoavg}),
with weights proportional to the reciprocal of the length of the 
corresponding edge,
\[n_i = GA\Big( v_{i-1}, v_{i}; \frac{d_{i-1}}{d_i + d_{i-1}}\Big).\]
In case $p_i$ is a boundary vertex, the normal is taken as that of the only
neighboring edge. 

Figure~\ref{fig:bottle_naivenorm} depicts a control polygon and different 
curves obtained from it by two modified schemes with initial normals computed by
the "naive method". For comparison the curves generated by the corresponding
linear schemes from the same initial control polygon are also shown.
Figure~\ref{fig:bottle_naivenorm}a demonstrates that the MLR algorithm
with $m=3$ preserves the shape of the control polygon more accurately than 
the corresponding LLR scheme.
In Figure~\ref{fig:bottle_naivenorm}b we see that the L4Pt scheme generates
a self intersecting curve while the curve of the M4Pt scheme is self intersection 
free and follows the shape of the initial polygon smoothly.

The proposed "naive method" determines intuitive initial normals, which can 
be modified later on, as is shown in the example 
of subsection~\ref{subsubsec:video}.

\section*{Acknowledgment}
The authors thank Prof. D. Cohen-Or for fruitful discussions, and the reviewers for their valuable comments. This work
was partially supported by Minkowski Minerva center at Tel-Aviv University.


\begin{thebibliography}{10}

\bibitem{cahore:13}
Thomas~J. Cashman, Kai Hormann, and Ulrich Reif.
\newblock Generalized {L}ane-{R}iesenfeld algorithms.
\newblock {\em Computer Aided Geometric Design}, 30(4), 2013.

\bibitem{jue:07p}
Pavel Chalmoviansk\'y and Bert J\"uttler.
\newblock A non--linear circle--preserving subdivision scheme.
\newblock {\em Adv.\ Comp.\ Math.}, 27:375--400, 2007.

\bibitem{Conte:1980:ENA:578374}
Samuel~Daniel Conte and Carl W.~De Boor.
\newblock {\em Elementary Numerical Analysis: An Algorithmic Approach}.
\newblock McGraw-Hill Higher Education, 3rd edition, 1980.

\bibitem{cd:11}
Costanca Conti and Nira Dyn.
\newblock Analysis of subdivision schemes for nets of functions by proximity
  and controllability.
\newblock {\em Journal of Computational and Applied Math}, 236:461--475, 2011.

\bibitem{dubucDeslauriers2}
Gilles Deslauriers and Serge Dubuc.
\newblock Symmetric iterative interpolation processes.
\newblock {\em Constractive Approximation.}, 5(1):49--68, 1989.
\newblock Fractal approximation.

\bibitem{dosa:05}
Neil~A. Dodgson and Malcolm~A. Sabin.
\newblock A circle-preserving variant of the four-point scheme.
\newblock {\em Mathematical Methods for Curves and Surfaces: Tromso 2004},
  pages 275--286, 2005.

\bibitem{df:02}
Nira Dyn and Elza Farkhi.
\newblock Spline subdivision schemes for compact sets-a survey.
\newblock {\em Serdica Math. J.}, 28 (4):349--360, 2002.

\bibitem{dgl:87}
Nira Dyn, John~A. Gregory, and David Levin.
\newblock A 4-point interpolatory subdivision scheme for curve design.
\newblock {\em Comput. Aided Geom. Design}, 4(4):257--268, 1987.

\bibitem{dl:02}
Nira Dyn and David Levin.
\newblock Subdivision schemes in geometric modelling.
\newblock {\em Acta Numerica}.

\bibitem{ds:16}
Nira Dyn and Nir Sharon.
\newblock A global approach to the refinement of manifold data.
\newblock {\em Mathematics of Computation}, 1:1--2, 2016.

\bibitem{ds:14}
Nira Dyn and Nir Sharon.
\newblock Manifold-valued subdivision schemes based on geodesic inductive
  averaging.
\newblock {\em Journal of Computational and Applied Mathematics}, 2016.

\bibitem{kd:13}
Shay Kels and Nira Dyn.
\newblock Subdivision schemes of sets and the approximation of set-valued
  functions in the symmetric difference metric.
\newblock {\em Foundations of Computational Mathematics}, 13(5):835--865, 2013.

\bibitem{lr:80}
Jeffrey~M. Lane and Richard~F. Riesenfeld.
\newblock A theoretical development for the computer generation and display of
  piecewise polynomial surfaces.
\newblock {\em IEEE Transactions on Pattern Analysis and Machine Intelligence},
  PAMI-2(1):35--46, 1980.

\bibitem{merrien:92}
Jean-Louis Merrien.
\newblock A family of hermite interpolants by bisection algorithms.
\newblock {\em Numerical Algorithms}, 2:187--200, June 1992.

\bibitem{donoho}
Inam~Ur Rahman, Iddo Drori, Victoria~C. Stodden, David~L. Donoho, and Peter
  Schr{\"o}der.
\newblock Multiscale representations for manifold-valued data.
\newblock {\em Multiscale Model. Simul.}, 4(4):1201--1232, 2005.

\bibitem{wd:05}
Johannes Wallner and Nira Dyn.
\newblock Convergence and ${C}^1$ analysis of subdivision schemes on manifolds
  by proximity.
\newblock {\em Computer Aided Geometric Design}, 22(7):593--622, 2005.

\end{thebibliography}

\end{document}